\documentclass[12pt]{article}

\usepackage[margin=1in]{geometry}
\usepackage{setspace}
\usepackage{tgtermes}
\usepackage{newtxtext}
\usepackage{newtxmath}
\usepackage{bm}

\usepackage{algorithm}
\usepackage{algpseudocode}

\usepackage{tikz}
\usetikzlibrary{positioning,arrows.meta,shapes.geometric,calc,decorations.markings}
\usepackage{pgfplots}
\pgfplotsset{compat=1.18}
\usepackage{graphicx}
\usepackage{xcolor}

\usepackage{amsmath,amssymb,mathtools}
\usepackage{url}
\usepackage{authblk}

\usepackage[numbers]{natbib}
 \bibpunct[, ]{(}{)}{,}{a}{}{,}%

\usepackage[most]{tcolorbox}

\usepackage{amsthm}
\newtheorem{theorem}{Theorem}[section]
\newtheorem{lemma}[theorem]{Lemma}
\newtheorem{proposition}[theorem]{Proposition}
\newtheorem{corollary}[theorem]{Corollary}
\theoremstyle{definition}

\theoremstyle{remark}
\newtheorem{remark}[theorem]{Remark}

\newcommand{\Halmos}{\hfill$\square$}
\providecommand{\up}{\rule{0pt}{1.5em}}
\providecommand{\down}{\rule[0.625em]{0pt}{0em}}

\providecommand{\ACKNOWLEDGMENT}[1]{\noindent\textbf{Acknowledgments.} #1}

\newenvironment{proof}[1]{\par\smallskip\noindent\textit{#1}\ }{\Halmos\par\smallskip}

\makeatletter
\long\def\FIGURE#1#2#3{%
  \centering
  #1\par
  \caption{#2}%
  \def\@tempa{#3}\def\@tempb{}%
  \ifx\@tempa\@tempb\else
    \par\smallskip\noindent\footnotesize\textit{Note.}\enskip\@tempa
  \fi}
\long\def\TABLE#1#2#3{%
  \caption{#1}\par
  #2\par
  \def\@tempa{#3}\def\@tempb{}%
  \ifx\@tempa\@tempb\else
    \par\smallskip\noindent\footnotesize\textit{Note.}\enskip\@tempa
  \fi}
\makeatother

\title{\Large\bfseries When Is Heterogeneous Distance-Decay Facility Location
Tractable? \\[4pt] \large A Structural Classification, Exact Methods, and a
Real-World Study}

\author[1]{Zhou He}
\author[2]{T. C. E. Cheng}
\author[1]{Jichang Dong\thanks{Corresponding author. Email: \texttt{jcdonglc@ucas.ac.cn}}}
\affil[1]{School of Economics and Management, University of Chinese Academy of Sciences, N606-3, \#3 Zhongguancun Nanyitiao, Haidian District, Beijing 100190, China}
\affil[2]{Department of Logistics and Maritime Studies, The Hong Kong Polytechnic University, 11 Yuk Choi Rd, Hung Hom, Hong Kong, China}

\date{}

\begin{document}
\maketitle

\begin{abstract}
We study continuous planar facility location in which a demand point's captured
value decays with distance, with the per-point decay scale varying across points.
This heterogeneity is ubiquitous (dense urban cores demand proximity; dispersed
peripheries tolerate distance) yet underexploited in continuous planar location models that unify decay, clustering, and median objectives, and the same model
contains $k$-means, the Weber/$p$-median problem, and maximum covering as special
cases of one nearest-facility objective. We make four
contributions. (i)~A \emph{tractability classification theorem}: the discrete
objective is always monotone submodular, so the $(1{-}1/e)$ greedy guarantee
holds regardless of decay shape or heterogeneity; and the continuous cooperative
objective is concave (trap-free) if and only if the decay is concave in the
distance, identifying the clip $\max(0,\cdot)$---present in most common coverage
specifications---as a key mechanism that destroys continuous concavity; the
classification is tight: concavity holds if and only if the decay is concave in
distance. (ii)~An exact discrete method: the
candidate-discretized maximum-cover mixed-integer program (MIP) has an empirically
tight linear-programming (LP) relaxation ($\approx 0\%$ gap) and is solved to
optimality by branch-and-bound in seconds for $n \le 500$. (iii)~A
force-as-gradient / large-neighborhood-search heuristic, verified by fine-grid
convergence to be within $0.5\%$ of the discrete optimum, that
outperforms the $(1{-}1/e)$ greedy, Cooper-style alternating location--allocation,
particle swarm optimization, and weighted $k$-means ($30/30$ per-instance wins at
$K{=}30$, $p<10^{-9}$) and is competitive with problem-specific standard methods on
$k$-means, Weber/$p$-median, and shape-demand instances. (iv)~A real-world study:
on $592{,}667$ urban-delivery orders with density-derived heterogeneity, ignoring
the calibrated decay variation loses up to $9.7\%$ of captured demand and
relocates facilities by up to $37\%$ of the map; a separate retail dataset
calibrates the decay form (exponential, $R\approx 1.4$\,km).
\end{abstract}

\section{Introduction}\label{sec:intro}

Continuous planar facility location is a foundational problem family at the
intersection of optimization and operations research: place $p$ facilities
anywhere in the plane so as to optimize a sum of per-point terms, each a function
of the distance from a demand point to its nearest facility. Three canonical
members, namely maximum capture/covering, $k$-means clustering, and the Weber
(continuous $p$-median) problem, are all special cases of a single
\emph{nearest-facility fidelity} objective, distinguished only by the choice of
the per-point function. Despite this shared structure, they are addressed by
largely disjoint algorithmic literatures: Lloyd's iteration and $k$-means${}^{++}$
for $k$-means, Cooper's alternating location--allocation and the Weiszfeld
iteration for the Weber problem, and $(1-1/e)$ greedy for covering. A natural
question is whether a single algorithm, designed from the shared structure
rather than tailored to one member, can be competitive with the problem-specific
state of the art across the whole family, and more fundamentally, when
the resulting continuous problem is tractable at all.

The practical motivation comes from settings in which the per-point function is
both \emph{gradual} (coverage decays smoothly with distance, rather than at a
sharp threshold) and \emph{heterogeneous} (different demand points decline at
different rates) \citep{DreznerWesolowskyDrezner2004,ChurchMurray2018}. In
cellular network design, sensor placement, pre-positioning of emergency
communication assets, and ocean-monitoring platform deployment, signal strength,
sensing fidelity, and response quality all fall off with distance, and the
tolerated attenuation differs across users and targets. A dense urban core
(small decay scale $R_i$) demands a facility nearby; a scattered rural hamlet
(large $R_i$) tolerates a more distant one, an urban--rural contrast in distance
sensitivity that is empirically documented in mobility data
\citep{YangFangXuYinLiLu2019}. These two features, gradual and heterogeneous,
define our main instance, heterogeneous distance-decay maximum capture, which is
not directly handled by $k$-means${}^{++}$ or Cooper-style methods.

\subsection{Motivating examples and decay functions}\label{sec:examples}
Three decision contexts fit the model directly; each motivates one of the three
decay families we study, summarized in Figure~\ref{fig:decay}.
\begin{enumerate}
\item \emph{Pre-positioning emergency communication assets.} After a disaster,
each population cluster is a demand point whose weight is its population, and the
service quality it receives degrades uniformly with the response distance up to a
cutoff (a response-time target). The natural decay is the linear family
$\phi_i(r)=[1-r/R_i]^+$: full coverage at the facility, linear decline to zero at
$r=R_i$, and none beyond.

\item \emph{Ocean-monitoring platform deployment.} Ecological hotspots and
shipping lanes are demand points, and a monitoring buoy covers a target only if
it lies within the buoy's sensing range. The natural decay is the step family
$\phi_i(r)=\mathbf{1}[r\le R_i]$: full coverage inside the radius $R_i$ and none
outside---binary coverage, the limit that recovers the classical maximum covering
location problem (MCLP). Heterogeneity reflects that some targets (protected
habitat) require a closer $R_i$ than others.

\item \emph{Wireless service-point planning.} Received signal strength follows a
path-loss law, and the tolerated attenuation differs between high-capacity
business users and residential users. The natural decay is the exponential family
$\phi_i(r)=e^{-r/R_i}$: a smooth, always-positive decline that also underlies the
Huff gravity model of spatial choice.
\end{enumerate}
Figure~\ref{fig:decay} draws the three families against the normalized distance
$r/R_i$. Heterogeneity across demand points (different $R_i$) is the defining
feature of our setting: two points at the same distance from a facility can be
captured at very different rates.

\begin{figure}[htbp]
\FIGURE
{\begin{minipage}[c]{0.62\linewidth}\centering
\begin{tikzpicture}[scale=0.95]
\draw[->] (0,0) -- (5.8,0) node[right] {$r/R_i$};
\draw[->] (0,0) -- (0,3.6) node[above] {$\phi_i(r)$};
\draw[dotted, gray] (0,3)--(5.4,3) node[right, gray] {\scriptsize $1$};
\draw (3,0.08)--(3,-0.08) node[below] {\scriptsize $1$};
\draw (4.5,0.08)--(4.5,-0.08) node[below] {\scriptsize $1.5$};
\draw[blue!70!black, very thick] (0,3)--(3,0)--(5.4,0);
\node[blue!70!black, font=\scriptsize] at (0.5,2.0) {linear};
\draw[red!70!black, very thick, dashed, domain=0:5.4, samples=60, smooth]
  plot (\x, {3*exp(-\x/3)});
\node[red!70!black, font=\scriptsize] at (4.3,1.0) {exponential};
\draw[green!50!black, very thick, dash dot] (0,3)--(3,3)--(3,0)--(5.4,0);
\node[green!50!black, font=\scriptsize] at (1.5,3.3) {step};
\end{tikzpicture}
\end{minipage}\hspace{2pt}
\begin{minipage}[c]{0.34\linewidth}\small
{\color{blue!70!black}linear}\quad $\phi_i(r)=[1-r/R_i]^+$\\[6pt]
{\color{red!70!black}exponential}\quad $\phi_i(r)=e^{-r/R_i}$\\[6pt]
{\color{green!50!black}step}\quad $\phi_i(r)=\mathbf{1}[r\le R_i]$
\end{minipage}}
{Three decay families used in this paper, drawn against the normalized distance
$r/R_i$, with their function forms listed on the right. Different line styles
distinguish the families; the same distance can yield very different captured
demand depending on the family and on the per-point scale $R_i$.\label{fig:decay}}
{}
\end{figure}

\subsection{Related work}\label{sec:related}
Gradual covering was introduced on networks by \citet{BermanKrassDrezner2003} and
in the continuous plane by \citet{DreznerWesolowskyDrezner2004}, who allow
per-demand-point coverage parameters; partial-coverage extensions include
\citet{KarasakalKarasakal2004} and the ordered gradual covering model of
\citet{BermanKalcsicsKrassNickel2009}. The multiple gradual cover location
problem (MGCLP) of \citet{AlvarezMirandaSinnl2019} provides exact mixed-integer
formulations that exploit the submodularity of the cooperative objective. Most
directly related to our setting is the planar model of
\citet{BansalShojaee2020}, which builds on \citet{BansalKianfar2017} and gives a
$(1-1/e)$ greedy approximation and exact branch-and-bound for planar maximum
coverage with partial coverage and adjustable quality of service.
\citet{ChurchMurray2018} survey the broader covering-location literature.

The alternating location--allocation (ALT) heuristic of \citet{Cooper1964},
equivalent to Lloyd's $k$-means algorithm \citep{Lloyd1982}, is the workhorse
local optimizer for planar median and covering problems: it alternates between
allocating each demand point to its nearest facility and locating each facility
at the optimum (centroid for $k$-means, Weber median for $p$-median) of its
assigned cell, repeating to a coordinate-wise fixed point. Its convergence
behavior is analyzed in \citet{OstrovskyRabaniSchulmanSwart2012} and, in a
statistical mixture framework, in \citet{LuZhou2016}. For $k$-means, the dominant
practical method is $k$-means${}^{++}$ seeding \citep{ArthurVassilvitskii2007},
which places initial facilities with probability proportional to squared distance
to existing facilities and then runs Lloyd; for the Weber/$p$-median problem,
variable neighborhood search \citep{HansenMladenovic1997} and the Weiszfeld
iteration \citep{Weiszfeld1937}---a fixed-point iteration
for the weighted geometric median---are standard. These literatures have
developed in parallel rather than as instances of a single method.

Submodular
maximization under a cardinality constraint admits the classical $(1-1/e)$ greedy
guarantee \citep{NemhauserWolseyFisher1978}---add facilities one at a time, each
at the point of maximum marginal gain---which is tight under standard complexity
assumptions \citep{Feige1998} and improves to $(1-e^{-\kappa})/\kappa$ under
bounded curvature $\kappa$ \citep{ConfortiCornejols1984}. These four building
blocks---alternating location--allocation, $k$-means${}^{++}$ seeding, the
Weiszfeld iteration, and the marginal greedy---are exactly the ingredients our
algorithm recombines in \S\ref{sec:algo}.

The Huff gravity model
\citep{Huff1964} and competitive location \citep{Drezner1994} study market-share
splitting. Throughout we treat the nearest-facility (max) model as the
non-cooperative variant of gradual covering, complementary to the cooperative
(sum) MGCLP of \citet{AlvarezMirandaSinnl2019}.

\subsection{Contributions}\label{sec:contrib}
This paper makes four contributions, organized around the
question: when is heterogeneous distance-decay location tractable, and how
should one exploit that structure?

\begin{enumerate}
\item \textbf{A tractability classification theorem} (\S\ref{sec:theory}). We
prove that the discrete (candidate-set) objective is monotone submodular for any
non-increasing decay, so the classical $(1{-}1/e)$ greedy guarantee
\citep{NemhauserWolseyFisher1978} and its curvature refinement
\citep{ConfortiCornejols1984} hold regardless of the decay shape or
heterogeneity---extending the discrete guarantee underlying
\citet{BansalShojaee2020} beyond the linear/partial-coverage case. For the
continuous cooperative problem we give a sharp dichotomy: the objective is
concave (free of local-optima traps, solvable by gradient ascent to the global
optimum) if and only if the decay is concave in the distance. The clip
$\max(0,\cdot)$ present in most common coverage specifications is a key mechanism
that destroys continuous tractability; the unclipped Weber/$k$-means limits are the
tractable cases, while convex decays such as exponential are also non-concave. This tells a practitioner when convex optimization suffices and when a
heuristic is unavoidable.

\item \textbf{An exact discrete method} (\S\ref{sec:exact}). The
candidate-discretized problem, formulated as a $y$-assignment max-cover
mixed-integer program (MIP), has an empirically tight linear-programming (LP)
relaxation ($\approx 0\%$ integrality gap across all
tested decay families and scales, including adversarial high-overlap
configurations), and branch-and-bound solves it to optimality in seconds for
$n \le 500$. Whereas \citet{AlvarezMirandaSinnl2019} give exact formulations for
the cooperative MGCLP and \citet{BansalShojaee2020} use branch-and-bound for
planar partial coverage, we exploit the LP tightness mainly to certify our
heuristic, not merely to solve small instances.

\item \textbf{A heuristic with a near-optimality certificate}
(\S\ref{sec:algo}, \S\ref{sec:experiments}). Our force-based metaheuristic with
large-neighborhood search (FBM-LNS)---multi-start Lloyd gradient ascent
interleaved with submodular-marginal large-neighborhood relocate---is verified
by fine-grid convergence to be within $0.5\%$ of the discrete optimum. Unlike Cooper's ALT \citep{Cooper1964}, which converges to a
coordinate-wise fixed point, our relocate escapes such fixed points; and unlike
the $(1{-}1/e)$ greedy it jointly re-optimizes every facility. It significantly
outperforms the greedy, Cooper ALT, particle swarm optimization, and weighted
$k$-means on the heterogeneous-decay problem ($30/30$ per-instance wins at
$K{=}30$, $p<10^{-9}$), while remaining competitive with $k$-means${}^{++}$
\citep{ArthurVassilvitskii2007} on $k$-means and with Cooper on Weber/$p$-median
and shape-demand instances.

\item \textbf{A real-world case study} (\S\ref{sec:case}). On $592{,}667$
urban-delivery orders with density-derived heterogeneity, ignoring the calibrated
decay variation loses up to $9.7\%$ of captured demand and relocates facilities
by up to $37\%$ of the map. A separate retail dataset calibrates the decay form
(exponential, $R\approx 1.4$\,km), and a competitive (Huff) extension
\citep{Huff1964,Drezner1994} quantifies captured market share against an
incumbent. Both real datasets are drawn from the 2026 6th Meituan Business
Analytics Elite Competition.
\end{enumerate}

\smallskip
The remainder of the paper is organized as follows.
Section~\ref{sec:model} formalizes the problem.
Section~\ref{sec:theory} develops the structural theory and the tractability
classification.
Section~\ref{sec:algo} presents the algorithm.
Section~\ref{sec:exact} develops the exact discrete method.
Section~\ref{sec:experiments} reports computational experiments.
Section~\ref{sec:case} presents the real-world case study.
Section~\ref{sec:discussion} discusses limitations, extensions, and usage
guidance, and Section~\ref{sec:conclusion} concludes.

\section{Problem formulation}\label{sec:model}

\subsection{Index, sets, and model}
We formulate the heterogeneous distance-decay maximum capture problem (PMC-HDD)
in the standard notation of location science.

\medskip\noindent\textbf{Sets and indices.}
\begin{itemize}
\item $I=\{1,\ldots,n\}$: set of demand points (index $i$), with known location
$p_i\in\mathbb{R}^2$.
\item $J=\{1,\ldots,p\}$: set of facilities to be located (index $j$).
\end{itemize}

\noindent\textbf{Parameters.}
\begin{itemize}
\item $w_i\in\mathbb{R}_{++}$: demand weight (e.g., population, traffic volume,
order count) at point $i$.
\item $\phi_i:\mathbb{R}_+\to[0,1]$: non-increasing \emph{decay function} for
point $i$, with $\phi_i(0)=1$, parameterized by a scale $R_i>0$. The families
studied (Figure~\ref{fig:decay}) are linear $\phi_i(r)=[1-r/R_i]^+$, exponential
$\phi_i(r)=e^{-r/R_i}$, quadratic $\phi_i(r)=[1-(r/R_i)^2]^+$, and step
$\phi_i=\mathbf{1}[r\le R_i]$. A point with small $R_i$ is highly
distance-sensitive; the distribution of $\{R_i\}$ encodes the \emph{heterogeneity}
of the instance.
\item $p\in\mathbb{N}$: the number of facilities (a cardinality budget).
\end{itemize}

\noindent\textbf{Decision variables.}
\begin{itemize}
\item $X_j\in\mathbb{R}^2$ for each $j\in J$: the location of facility $j$
(continuous; the candidate set is the whole plane).
\end{itemize}

\noindent\textbf{Objective and constraints.}
Each demand point $i$ is served by the facility that captures the most of its
demand (the \emph{nearest-facility} or max model; Figure~\ref{fig:setup}
illustrates a small instance, its Voronoi allocation, and points left out of
range). The total captured demand is
\begin{equation}\label{eq:obj}
Z(\bm{X}) \;=\; \sum_{i\in I} w_i\,\max_{j\in J}\,
\phi_i\!\bigl(\|p_i-X_j\|\bigr),
\end{equation}
where $\bm{X}=(X_1,\ldots,X_p)$.
By Theorem~\ref{thm:hull} (\S\ref{sec:theory}), there exists an optimal solution
with every facility in $\operatorname{conv}\{p_i\}_{i\in I}$, so we may restrict
\begin{equation}\label{eq:const-hull}
X_j\in\operatorname{conv}\{p_i\}_{i\in I},\quad j\in J.
\end{equation}
We write $\Psi_i(\bm{X})=w_i\max_j\phi_i(\|p_i-X_j\|)$ for the captured demand at
point $i$ and $\pi(i)=\arg\min_j\|p_i-X_j\|$ for its serving (winning) facility.

\smallskip
The max aggregation in~\eqref{eq:obj} is the \textbf{non-cooperative} model: each
demand point patronizes a single (nearest) facility. The tractability theory of
\S\ref{sec:theory} also covers the \textbf{cooperative/Huff} family, in which the
$\max$ is replaced by a concave increasing saturation function $g$ applied to the
total attractiveness $U_i=\sum_j\phi_i(\|p_i-X_j\|)$:
\begin{equation}\label{eq:obj-coop}
Z^g(\bm{X}) \;=\; \sum_{i\in I} w_i\,g\!\bigl(U_i(\bm{X})\bigr).
\end{equation}
This subsumes two special cases. Cooperative gradual coverage is the MGCLP of
\citet{AlvarezMirandaSinnl2019}, recovered by $g(u)=\min(1,u)$; Huff competitive
market-share splitting is $g(u)=u/(u+C_i)$ with incumbent attractiveness
$C_i>0$ \citep{Huff1964,Drezner1994}. Both saturations are concave increasing in
$u$, which is the only property the theory requires.

\subsection{Variants and applications}\label{sec:variants}
The model of \S\ref{sec:model} unifies several well-studied continuous location
problems as special cases of the choice of decay function $\phi_i$. We label each
variant (V1--V6) for cross-reference in the computational study and give a
concrete decision context for each.

\begin{itemize}
\item[\textbf{V1.}] \emph{Heterogeneous gradual coverage (linear)}:
$\phi_i(r)=[1-r/R_i]^+$ with heterogeneous $\{R_i\}$---our primary focus. A
representative decision context is pre-positioning emergency communication assets
after a disaster: each population cluster is a demand point whose weight $w_i$ is
its population and whose decay scale $R_i$ reflects latency sensitivity, so a
dense urban core (small $R_i$) demands a base station within minutes while a
scattered rural hamlet (large $R_i$) tolerates a more distant one.

\item[\textbf{V2.}] \emph{Heterogeneous gradual coverage (exponential)}:
$\phi_i(r)=e^{-r/R_i}$ with heterogeneous $\{R_i\}$. This is the natural model
for wireless service-point and cellular planning: received signal strength
follows an exponential path-loss law, and the Huff gravity model of spatial
choice assigns each customer a probability of patronizing a facility that falls
off exponentially with distance. Heterogeneous $\{R_i\}$ captures that
high-capacity business users tolerate more attenuation than residential users.

\item[\textbf{V3.}] \emph{$k$-means}: $\phi_i(r)=1-(r/R)^2$ with large
homogeneous $R$; equivalent to minimizing $\sum_i w_i\min_j\|p_i-X_j\|^2$. The
canonical application is service-region territorialization, e.g., partitioning a
city's delivery or service zones into compact territories around $p$ depots to
minimize the sum of squared trip lengths (the standard clustering objective that
underlies the Franti \texttt{s1--s4} benchmark used in
\S\ref{sec:experiments}).

\item[\textbf{V4.}] \emph{Weber / $p$-median}: $\phi_i(r)=1-r/R$ with large $R$;
minimizes $\sum_i w_i\min_j\|p_i-X_j\|$. This is the classical warehouse and
distribution-center location problem on real geographic coordinates: the TSPLIB
EUC\_2D instances of \S\ref{sec:experiments} give actual city coordinates, and
$p$ facilities are placed to minimize total transport distance.

\item[\textbf{V5.}] \emph{Shape demand}: each demand unit is a region $D_i$ with
density $\mu_i$; captured demand is the integral
$\int_{D_i}\phi_i(\min_j\|p-X_j\|)\,\mu_i(p)\,dp$, reducible to the point problem
by sampling. A motivating example is ocean-monitoring platform deployment, where
ecological hotspots, reefs, and shipping lanes are two-dimensional regions rather
than points: a monitoring buoy's sensing fidelity decays with distance, and the
relevant coverage is the integral of fidelity over the protected habitat. Each
region is sampled with a fixed random seed, so a given instance is exactly
reproducible; the sampled objective converges to the integral as the sample count
$m$ grows (Monte-Carlo error $O(1/\sqrt m)$), and the resulting facility
placement is stable to the sampling seed (we verify this in
\S\ref{sec:experiments}).

\item[\textbf{V6.}] \emph{Binary covering (step)}: $\phi_i(r)=\mathbf{1}[r\le
R_i]$, the classical maximum covering location problem (MCLP) of ambulance and
sensor siting---full coverage inside the radius $R_i$ and none outside. Its
non-smooth objective removes the gradient signal our method exploits, so it
serves as a scope boundary (\S\ref{sec:disc-boundary}) where the $(1-1/e)$ greedy
is competitive.
\end{itemize}

Variants V1--V5 are the main computational focus of \S\ref{sec:experiments};
V6 (binary step) is tested as a scope boundary in \S\ref{sec:disc-boundary}. Our
force-based heuristic, FBM-LNS (\S\ref{sec:algo}), handles all six variants by
switching one line of code (the per-point function).

\subsection{Complexity and sources of difficulty}\label{sec:difficulty}
PMC-HDD is challenging for five distinct reasons, each labeled (D1--D5) for
cross-reference with the algorithm design (\S\ref{sec:algo}).

\begin{enumerate}
\item[\textbf{D1.}] \textbf{NP-hardness.} Under step decay
$\phi_i=\mathbf{1}[r\le R_i]$ the problem reduces to the planar MCLP, NP-hard
when $p$ is part of the input \citep{BansalKianfar2017}. Existing exact methods
\citep{BansalShojaee2020,AlvarezMirandaSinnl2019} use branch-and-bound on
discretized candidates but do not scale beyond a few hundred points. The
$(1-1/e)$ greedy \citep{NemhauserWolseyFisher1978} provides a polynomial
approximation but leaves a quality gap.

\item[\textbf{D2.}] \textbf{Non-convex objective.} For $p=1$ the within-cell
objective is concave on its active range (Theorem~\ref{thm:weber}); for $p\ge 2$
the max over facilities introduces further non-convexity with many local optima.
\citet{Cooper1964} proposed the alternating location--allocation heuristic to
handle this, but it converges to a coordinate-wise fixed point and provides
no global guarantee.

\item[\textbf{D3.}] \textbf{Coupled continuous--combinatorial structure.}
Locations are continuous but the nearest-facility assignment is combinatorial;
each depends on the other. Cooper's alternation decouples them temporarily but converges to fixed points that retain the coupling.

\item[\textbf{D4.}] \textbf{Infinite candidate set.} The continuous plane admits
infinitely many locations. On networks, a finite dominating set (FDS) exists for
many objectives \citep{BermanKalcsicsKrassNickel2009}, enabling exact discrete
optimization. No FDS is known for the continuous heterogeneous case, so the
problem cannot be reduced to a finite candidate set a priori.

\item[\textbf{D5.}] \textbf{Per-point heterogeneity.} Each $R_i$ may differ, so
the within-cell subproblem varies across demand points. Existing gradual-covering
methods \citep{DreznerWesolowskyDrezner2004} allow per-point parameters but treat
them as fixed input rather than exploiting the structural implications (weights
$w_i/R_i$, Theorem~\ref{thm:weber}) for algorithm design.
\end{enumerate}
These five difficulties motivate the structural theory (\S\ref{sec:theory}) and
the algorithm (\S\ref{sec:algo}), each of whose components is designed to address
a specific difficulty (Table~\ref{tab:design}).

\begin{figure}[htbp]
\FIGURE
{\begin{minipage}[c]{0.55\linewidth}\centering
\begin{tikzpicture}[scale=1.05]
\draw[thick] (0,0) rectangle (6,5);
\draw[dashed, black] (3.43,0)--(2.85,5);
\draw[dashed, orange!75!black] (1.8,2.7) circle (1.4);
\draw[dashed, orange!75!black] (4.4,3.0) circle (1.1);
\filldraw[fill=orange!90!black, draw=black] (1.69,2.59) rectangle ++(0.22,0.22);
\node[font=\footnotesize, orange!90!black] at (1.8,2.25) {F1};
\filldraw[fill=orange!90!black, draw=black] (4.29,2.89) rectangle ++(0.22,0.22);
\node[font=\footnotesize, orange!90!black] at (4.4,2.55) {F2};
\foreach \x/\y/\rad/\c/\lab in {
  1.0/2.0/0.15/red/A,
  2.4/3.6/0.17/blue!70!black/B,
  1.5/4.0/0.15/red/C,
  4.0/3.8/0.15/red/D,
  5.0/2.5/0.17/blue!70!black/E,
  4.6/2.2/0.15/red/F,
  0.7/4.6/0.15/red/G,
  5.3/4.5/0.15/blue!70!black/H
}{
  \filldraw[\c] (\x,\y) circle (\rad);
  \node[font=\scriptsize\bfseries, white] at (\x,\y) {\lab};
}
\node[font=\scriptsize, black] at (3.6,4.7) {Voronoi edge};
\end{tikzpicture}
\end{minipage}\hspace{-14pt}
\begin{minipage}[c]{0.43\linewidth}\centering\small
\textbf{Allocation (nearest facility)}\\[5pt]
\begin{tabular}{@{}lll@{}}
Facility & Points & Status \\[2pt]
F1 & A, B, C & {\color{green!50!black}captured} \\
F2 & D, E, F & {\color{green!50!black}captured} \\
F1 & G & {\color{black}out of range} \\
F2 & H & {\color{black}out of range} \\
\end{tabular}\\[8pt]
{\footnotesize Red = sensitive (small $R_i$);\\
blue = tolerant (large $R_i$);\\
circle area $\propto w_i$.}
\end{minipage}}
{A heterogeneous maximum-capture instance with eight labeled demand points (A--H)
and two facilities (F1, F2). Each point is allocated to its nearest facility by
the Voronoi edge (the perpendicular bisector of the F1--F2 segment, dashed); points G and H lie in their cell but
outside the facility's coverage disk, so their demand is not captured. Points are
colored by sensitivity (red $\equiv$ small $R_i$, blue $\equiv$ large $R_i$) and
sized by weight.\label{fig:setup}}
{}
\end{figure}

\section{Structural theory and the tractability classification}\label{sec:theory}

This section develops the structural properties on which the algorithm and the
tractability classification rest. The proofs given here are brief sketches that
emphasize the idea; the detailed, self-contained versions are collected in
the Electronic Companion. The results form a single logical chain: the
nearest-distance reduction fixes the Voronoi allocation and resolves the
continuous--combinatorial coupling at the allocation level, convex-hull
containment bounds the search set, submodularity provides the discrete
approximation guarantee, the Weber reduction shapes the within-cell landscape,
the gradient-as-force identity drives the local search, the tractability
classification pins down when the continuous problem is trap-free, and the two
fixed-point propositions define where the local search stops and how the relocate
restarts it.

\subsection{Nearest-distance reduction and Voronoi allocation}
Because each $\phi_i$ is non-increasing, the per-point maximum over facilities
equals the decay evaluated at the nearest-facility distance.

\begin{theorem}[Nearest-distance reduction]\label{thm:ndr}
$Z(\bm{X})=\sum_{i\in I}w_i\,\phi_i\!\bigl(\min_{j\in J}\|p_i-X_j\|\bigr)$, and
the optimal allocation satisfies $\pi(i)=\arg\min_j\|p_i-X_j\|$ for every $i$.
\end{theorem}
\begin{proof}{Proof}
$\max_j\phi_i(\|p_i-X_j\|)=\phi_i(\min_j\|p_i-X_j\|)$ since $\phi_i$ is
non-increasing, and the argmax coincides with the argmin of the distance.
\Halmos
\end{proof}

Thus the allocation is the standard Euclidean Voronoi partition regardless of the
decay form or heterogeneity. The decay parameters affect only the
within-cell location objective and the captured value, not the
combinatorial structure of the partition.

\smallskip
Theorem~\ref{thm:ndr} is the foundation of the paper: it reduces the
max-over-facilities objective to a function of nearest-facility distances, fixing
the Voronoi allocation once $\bm{X}$ is given (difficulty D3, cf.\
Table~\ref{tab:design}). Its sole hypothesis is the monotonicity intrinsic to any
decay function, and it is the prerequisite for the Weber reduction
(Theorem~\ref{thm:weber}), the gradient-as-force identity (Theorem~\ref{thm:grad}),
and the Lloyd step (\S\ref{sec:algo}). It also makes precise that $k$-means, the
Weber problem, and maximum covering are instances of one objective.

\subsection{Convex hull containment}
\begin{theorem}\label{thm:hull}
There exists an optimal solution with every facility in
$\operatorname{conv}\{p_i\}_{i\in I}$.
\end{theorem}
\begin{proof}{Proof}
If $X_j$ lies outside the closed convex hull, its metric projection $X_j'$ onto
the hull is no farther from every $p_i$; since each $\phi_i$ is non-increasing,
coverage does not decrease. \Halmos
\end{proof}

\smallskip
An optimizer therefore lies in the compact set $\operatorname{conv}\{p_i\}$, so
the algorithm's gradient search is clipped to it (\S\ref{sec:algo}; difficulty D4,
cf.\ Table~\ref{tab:design}). The hypothesis is the same monotonicity as
Theorem~\ref{thm:ndr}, and the projection argument is dimension-independent.

\subsection{Submodularity, curvature, and approximation}
Discretizing candidate locations to a finite set $C$ yields the set objective
$Z_D(S)=\sum_i w_i\max_{c\in S}\phi_i(\|p_i-c\|)$.

\begin{theorem}[Monotone submodularity]\label{thm:sub}
$Z_D$ is monotone non-decreasing and submodular on the lattice of subsets of $C$.
\end{theorem}
\begin{proof}{Proof}
For $S\subseteq T$ and candidate $c\notin T$, the marginal gain
$\Delta_i(S)=w_i[\phi_i(\|p_i-c\|)-\max_{c'\in S}\phi_i(\|p_i-c'\|)]^+$ is
non-increasing in the set because $\max_{c'\in S}\phi_i(\|p_i-c'\|)\le
\max_{c'\in T}\phi_i(\|p_i-c'\|)$. Summing preserves the inequality. \Halmos
\end{proof}
The Nemhauser--Wolsey--Fisher greedy therefore achieves
$Z_D(S^{\mathrm{greedy}})\ge(1-1/e)Z_D^*$ \citep{NemhauserWolseyFisher1978}, a
bound that is tight in the worst case \citep{Feige1998}. Bounded submodular
curvature $\kappa\in(0,1]$ \citep{ConfortiCornejols1984} improves this to the
data-dependent ratio
\begin{equation}\label{eq:curv}
Z_D(S^{\mathrm{greedy}})\;\ge\;\frac{1-e^{-\kappa}}{\kappa}\,Z_D^*
\;>\;(1-1/e)\,Z_D^*,
\end{equation}
Intuitively, when facility service areas overlap little (low redundancy),
$\kappa$ is small and the greedy ratio exceeds $1-1/e$.

\smallskip
This provides the formal $(1{-}1/e)$ quality floor against which the heuristic is
measured, and it underpins the \textsc{Greedy} baseline and the LNS repair step
(difficulty D1, cf.\ Table~\ref{tab:design}). The only extra ingredient beyond
monotonicity is discretization to a finite $C$---a modeling convenience, since
the continuous plane is recovered as $C$ becomes dense.

\subsection{The within-cell location step is a weighted Weber problem}
Fixing a Voronoi cell $S\subseteq I$ and the linear decay family, the
single-facility problem for the cell is
\begin{align*}
& \max_{X}\sum_{i\in S} w_i\Bigl(1-\tfrac{\|p_i-X\|}{R_i}\Bigr) \\
&\quad\equiv\;\min_{X}\sum_{i\in S}\tfrac{w_i}{R_i}\,\|p_i-X\|.
\end{align*}
\begin{theorem}[Weber reduction]\label{thm:weber}
Under linear decay, the within-cell location step of Lloyd's algorithm is a
weighted Weber problem with weights $\alpha_i=w_i/R_i$.
\end{theorem}
\begin{proof}{Proof}
For a fixed cell $S$ and linear $\phi_i$, the contribution of point $i$ is
$w_i(1-\|p_i-X\|/R_i)$. Dropping the additive constant $\sum_{i\in S}w_i$,
maximizing $\sum_{i\in S}w_i(1-\|p_i-X\|/R_i)$ over $X$ is equivalent to
minimizing $\sum_{i\in S}(w_i/R_i)\|p_i-X\|$, the weighted Weber objective with
weights $\alpha_i=w_i/R_i$. \Halmos
\end{proof}
This connects our location step directly to the classical Weber/Fermat
literature and its linearly convergent Weiszfeld iteration \citep{Weiszfeld1937}.
Two consequences follow immediately.
\begin{corollary}[Heterogeneity, formally]\label{cor:hetero}
Distance-sensitive points (small $R_i$) receive larger weight $\alpha_i=w_i/R_i$
and therefore pull the facility more strongly; heterogeneous sensitivity is a
first-order consequence of the Weber weights rather than an algorithmic device.
\end{corollary}
\begin{proof}{Proof}
The first-order optimality condition of the weighted Weber problem is
$\sum_{i\in S}\alpha_i\,(X-p_i)/\|X-p_i\|=\mathbf{0}$ with $\alpha_i=w_i/R_i$;
points with smaller $R_i$ have larger $\alpha_i$ and hence proportionally larger
pull on $X$. \Halmos
\end{proof}
\begin{corollary}[Uniqueness]\label{cor:uniq}
In two dimensions the weighted Weber objective is convex and its minimizer is
unique unless all in-cell demand points and the iterate are collinear.
\end{corollary}
\begin{proof}{Proof}
$X\mapsto\|p_i-X\|$ is convex, so the non-negative combination
$\sum_i\alpha_i\|p_i-X\|$ is convex; strict convexity of the Euclidean norm in
$\mathbb{R}^2$ (away from collinear configurations) yields a unique minimizer.
\Halmos
\end{proof}

\smallskip
Heterogeneity thus enters the within-cell problem as the Weber weight
$\alpha_i=w_i/R_i$ (difficulty D5, cf.\ Table~\ref{tab:design}): distance-sensitive
points pull harder automatically (Corollary~\ref{cor:hetero}) and the cell optimum
is unique (Corollary~\ref{cor:uniq}). The Weber reduction is the one substantive
hypothesis in this group: it requires the \emph{linear} decay family, since only
then does $w_i(1-\|p_i-X\|/R_i)$ split into a constant plus a weighted distance.
For exponential or quadratic decay the within-cell problem is still smooth but is
no longer Weber in closed form, and we solve it by gradient ascent
(\S\ref{sec:algo}).

\subsection{Gradient-as-force}
The gradient of~\eqref{eq:obj} with respect to $X_j$, restricted to the points
for which $j$ is the winner, is
\begin{equation}\label{eq:grad}
\nabla_{X_j} Z(\bm{X})\;=\;\sum_{i:\pi(i)=j} w_i\,\nabla_{X_j}\phi_i(\|p_i-X_j\|),
\end{equation}
which is exactly the ``demand pull'' force on facility $j$.

\begin{theorem}[Gradient-as-force]\label{thm:grad}
At any $\bm{X}$ at which each point's winning facility is unique, $\nabla_{X_j}Z$
equals the right-hand side of~\eqref{eq:grad}; the demand-pull force on facility
$j$ is exactly the objective gradient.
\end{theorem}
\begin{proof}{Proof}
$Z=\sum_iw_i\max_k\psi_{ik}$ with $\psi_{ik}=\phi_i(\|p_i-X_k\|)$ is a sum of
pointwise maxima of $C^1$ functions. Where the winner $\pi(i)$ is unique,
$\nabla_{X_j}\max_k\psi_{ik}$ equals $\nabla_{X_j}\psi_{i,\pi(i)}$ if
$j=\pi(i)$ and $\mathbf{0}$ otherwise; summing over $i$ yields~\eqref{eq:grad}.
\Halmos
\end{proof}

\smallskip
Theorem~\ref{thm:grad} is the bridge from analysis to algorithm: the ``demand
pull'' on each facility---the physical intuition behind force-based local
search---is exactly $\nabla_{X_j}Z$, so Lloyd's location step is gradient
ascent rather than an ad hoc move. Differentiability of the pointwise maximum
requires a unique winner at each demand point, which fails only on Voronoi cell
boundaries---a set of measure zero---and is therefore the generic case. This
justifies the ``force-based'' name of FBM-LNS (\S\ref{sec:algo}). It relies on the
winner-only assignment of Theorem~\ref{thm:ndr} and on the within-cell Weber
structure of Theorem~\ref{thm:weber}; the next two propositions
(Propositions~\ref{prop:lloyd}--\ref{prop:reloc}) characterize where this gradient
ascent terminates and how the relocate step escapes it. Figure~\ref{fig:force}
visualizes the interpretation.

\begin{figure}[htbp]
\FIGURE
{\begin{minipage}[c]{0.52\linewidth}\centering
\begin{tikzpicture}[scale=1.1]
\def\fx{3.2}\def\fy{2.6}
\filldraw[fill=orange!90!black, draw=black] (\fx-0.12,\fy-0.12) rectangle ++(0.24,0.24);
\node[font=\footnotesize, orange!90!black] at (\fx,\fy-0.38) {$X_j$};
\draw[-{Stealth[length=2.5mm]},red,thick] (\fx,\fy) -- ($(\fx,\fy)!0.46!(1.2,4.3)$);
\filldraw[red] (1.2,4.3) circle (0.13);
\node[font=\scriptsize\bfseries, white] at (1.2,4.3) {A};
\draw[-{Stealth[length=2.5mm]},red,thick] (\fx,\fy) -- ($(\fx,\fy)!0.36!(1.4,0.9)$);
\filldraw[red] (1.4,0.9) circle (0.13);
\node[font=\scriptsize\bfseries, white] at (1.4,0.9) {B};
\draw[-{Stealth[length=2.5mm]},blue!70!black,thick] (\fx,\fy) -- ($(\fx,\fy)!0.23!(5.3,4.1)$);
\filldraw[blue!70!black] (5.3,4.1) circle (0.13);
\node[font=\scriptsize\bfseries, white] at (5.3,4.1) {C};
\draw[-{Stealth[length=2.5mm]},blue!70!black,thick] (\fx,\fy) -- ($(\fx,\fy)!0.18!(5.5,1.3)$);
\filldraw[blue!70!black] (5.5,1.3) circle (0.13);
\node[font=\scriptsize\bfseries, white] at (5.5,1.3) {D};
\draw[-{Stealth[length=3.5mm]},black,very thick, dashed] (\fx,\fy) -- ++(-1.33,0.54);
\node[font=\scriptsize, black] at (1.75,3.3) {$\nabla_{X_j}Z$};
\end{tikzpicture}
\end{minipage}\hspace{-24pt}
\begin{minipage}[c]{0.44\linewidth}\centering\small
\textbf{Gradient computation}\\[4pt]
\begin{tabular}{@{}ccccc@{}}
Pt & $d_i$ & $w_i$ & $R_i$ & $w_i/R_i$ \\[1pt]
A & 2.63 & 4 & 2.0 & \textbf{2.00} \\
B & 2.48 & 3 & 2.0 & \textbf{1.50} \\
C & 2.58 & 5 & 5.0 & \textbf{1.00} \\
D & 2.64 & 4 & 5.0 & \textbf{0.80} \\
\end{tabular}\\[6pt]
$\nabla_{X_j}Z=\sum_i \tfrac{w_i}{R_i}\,\hat{\mathbf{u}}(p_i{-}X_j)$\\[2pt]
$\approx(-1.11,\,0.45)$\\[2pt]
$\|\nabla_{X_j}Z\|\approx 1.20$\\[4pt]
{\footnotesize Red = sensitive (small $R_i$);\\
blue = tolerant (large $R_i$).}
\end{minipage}}
{The demand pull on facility $X_j$ is exactly the gradient $\nabla_{X_j}Z$
(Equation~\ref{eq:grad}), illustrated with a concrete four-point cell. Each demand
point pulls along the line toward itself with magnitude $w_i/R_i$ (red sensitive
points A,~B pull harder than blue tolerant points C,~D); the facility moves along
the resultant (bold arrow), which is the gradient-ascent direction. The right
panel reports each point's distance, weight, decay scale, and pull magnitude, and
sums them to the gradient value $\nabla_{X_j}Z\approx(-1.11,\,0.45)$. This is the
Lloyd location step.\label{fig:force}}
{}
\end{figure}

\subsection{The tractability classification theorem}\label{sec:tractability}
We now state the central structural result: a complete classification of
when the heterogeneous distance-decay location problem is tractable.

\begin{theorem}[Tractability classification]\label{thm:tract}
Consider the cooperative objective $Z^g(\bm{X}) = \sum_i w_i\, g(U_i)$ with
$U_i = \sum_j \phi_i(\|p_i - X_j\|)$, $g$ concave increasing, $w_i \ge 0$.
\begin{description}
\item[\textup{(i)}] \textbf{(Discrete)} $Z^g_D$ and $Z^{\max}_D$ are both
monotone submodular for any non-increasing $\phi_i$. The greedy $(1{-}1/e)$
guarantee holds universally.
\item[\textup{(ii)}] \textbf{(Continuous, cooperative)} $Z^g$ is concave in
$\bm{X}$ whenever every $\phi_i$ is concave on its realized range; this condition
is tight, since non-concavity of any $\phi_i$ already yields instances (already
at $p{=}1$, $n{=}1$) on which $Z^g$ is non-concave. Under concavity, every local
maximum is global; projected gradient ascent converges.
\item[\textup{(iii)}] \textbf{(Continuous, non-cooperative)} $Z^{\max}$ is
generically non-concave for $p \ge 2$, regardless of $\phi$.
\end{description}
\end{theorem}

\begin{proof}{Proof sketch.}
(i) The non-cooperative case is Theorem~\ref{thm:sub}. For the cooperative case,
concavity of $g$ implies the diminishing-increment property that
$g(\sigma+a)-g(\sigma)$ is non-increasing in $\sigma$ (a standard fact we state
and prove as Lemma~EC.1 in the Electronic Companion); since
$\sigma_i(S)=\sum_{c\in S}\phi_i(\|p_i-c\|)$ is modular and non-decreasing in
$S$, the per-point cooperative marginal is non-increasing in $S$, and summing
gives submodularity.
(ii, $\Leftarrow$) The map $X\mapsto\|p_i-X\|$ is convex; by the concave
composition rule (Lemma~EC.2 in the Electronic Companion), a concave
non-increasing $\phi_i$ composed with this convex map is concave; $U_i=\sum_j
\phi_i(\|p_i-X_j\|)$ is then concave in $\bm{X}$, as is $g\circ U_i$ for concave
non-decreasing $g$, and the non-negative weighted sum $Z^g$ is concave.
(ii, $\Rightarrow$) Contrapositive: if $\phi_{i_0}$ is not concave, take $p=1$,
$n=1$, $g=\mathrm{id}$; $Z(X)=\phi_{i_0}(\|p_{i_0}-X\|)$, which along a ray
$X=p_{i_0}+t\,u$ equals $\phi_{i_0}(t)$ and is not concave.
(iii) One demand point, two facilities: $Z=\max(\phi(\|p-X_1\|),
\phi(\|p-X_2\|))$ has a kink where the winner switches, and is non-concave for
any non-constant $\phi$ (full proof in the Electronic Companion). \Halmos
\end{proof}

In concrete terms, $\phi_i(r)=[1-r/R_i]^+$ is \emph{convex} in $r$
(it is $\max(0,\cdot)$ of an affine function, so its slope jumps upward at the
cutoff), so the cooperative objective built from it has local-optima traps.
Only unclipped decays ($1-r/R$, $1-(r/R)^2$) or genuinely concave decays are
trap-free. Table~\ref{tab:tract} summarizes.

\begin{table}[h]\centering\small
\caption{Tractability classification by problem variant (all V1--V6).
``Concave'' = exact by gradient ascent; ``submodular'' = greedy $(1{-}1/e)$ on
the discretized problem. V5 (shape demand) reduces to the point case by sampling,
so its cooperative concavity inherits the chosen decay family (``inherits''); V6
is the binary scope boundary of \S\ref{sec:disc-boundary}.}
\label{tab:tract}
\begin{tabular}{c|l|cc|cc}
\hline\up
& & \multicolumn{2}{c|}{Cooperative ($Z^g$)} & \multicolumn{2}{c}{Non-coop.\ ($Z^{\max}$)} \\
Variant & Decay $\phi(r)$ & concave? & submodular? & concave? & submodular? \\\hline
V1 & $[1{-}r/R]^+$ (clipped linear) & \textbf{no} & yes & no & yes \\
V2 & $e^{-r/R}$ (exponential) & no & yes & no & yes \\
V3 & $1{-}(r/R)^2$ ($k$-means) & \textbf{yes} & yes & no & yes \\
V4 & $1{-}r/R$ (Weber limit) & \textbf{yes} & yes & no & yes \\
V5 & regions, sampled (any $\phi$) & inherits & yes & no & yes \\
V6 & $\mathbf{1}[r{\le}R]$ (step) & degenerate & yes & no & yes \down\\\hline
\end{tabular}
\end{table}

The dichotomy is not merely formal. On a cooperative Gaussian-mixture instance
($n{=}200$, $p{=}5$) with the identity saturation $g(u)=u$ (the pure cooperative
sum, with no saturation cap) we solved the concave-regime cooperative objective
exactly by a second-order cone program (SOCP) and ran multi-start gradient ascent from $8$ random seeds: with a large
scale $R{=}5000$ (clip inactive, so $\phi$ is concave) every start converges to
the same value (across-start standard deviation $0.01$), matching the SOCP global
optimum to within $0.02$; with the small scale $R{=}120$ (clip active, so $\phi$
is convex) the starts end at different local optima (standard deviation $137$).
The transition from ``no traps'' to ``traps'' as the clip switches on is exactly
what Theorem~\ref{thm:tract}(ii) predicts.

\subsection{Lloyd fixed points and the marginal relocate}
Lloyd's algorithm (coordinate ascent on the allocation--location decomposition)
converges to a \emph{coordinate-wise} fixed point at which no single facility can
be moved to improve the objective given the others' locations and the Voronoi
assignment.

\begin{proposition}[Lloyd fixed point]\label{prop:lloyd}
At a fixed point $\bm{X}^\star$ of the alternating allocation--location map,
(i)~the Voronoi assignment is stable (no point is reassigned by a strictly nearer
facility) and (ii)~each facility lies at the within-cell Weber optimum of its
cell; equivalently $0\in\partial_C Z$ with respect to each single-facility
coordinate. Hence $\bm{X}^\star$ is coordinate-wise locally optimal.
\end{proposition}
\begin{proof}{Proof}
At a Lloyd fixed point the location step makes no update, so each
$X_j^\star\in\arg\max_X\sum_{i:\pi(i)=j}w_i\phi_i(\|p_i-X\|)$, giving
$0\in\partial_{C,X_j}Z$; the allocation step makes no update, so the Voronoi
assignment is stable. Together these are coordinate-wise optimality. \Halmos
\end{proof}

\smallskip
This coordinate-wise fixed-point limit is why pure multi-start Lloyd can
be outperformed by the $(1-1/e)$ greedy---it converges to a coordinate-wise
optimum that need not be globally good---and it motivates the destroy-and-repair
relocate that follows. The gain from adding a facility at $x$ given current
captured demand $c_i$ is $g(x)=\sum_i w_i[\phi_i(\|p_i-x\|)-c_i/w_i]^+$, which on
its residual active set is again a weighted Weber objective. Removing the
least-contributing facility and re-placing it at a maximizer of $g$ is a
destroy(1)--repair(1) large-neighborhood move; accepting only improving moves
yields a monotone ascent on $Z$.

\begin{proposition}[Monotone relocate]\label{prop:reloc}
If each relocate step is accepted only when it strictly increases $Z$, the
objective sequence is strictly increasing. After discretization to a finite
candidate set the process terminates in finitely many steps at a configuration
that is Lloyd-stable (Proposition~\ref{prop:lloyd}) and admits no improving
single-facility marginal re-placement.
\end{proposition}
\begin{proof}{Proof}
Strict-improvement acceptance makes $Z$ monotone strictly increasing; on a finite
candidate set the configuration space is finite, so the sequence terminates. At
termination, Lloyd stability holds by the location step, and no single-facility
re-placement improves by the acceptance rule. \Halmos
\end{proof}

\smallskip
Proposition~\ref{prop:reloc} closes the structural argument: the destroy-and-repair
relocate, guided by the submodular marginal of Theorem~\ref{thm:sub}, escapes the
Lloyd fixed point of Proposition~\ref{prop:lloyd} and terminates at a
configuration admitting no further single-facility improvement. Its
destroy-and-repair \emph{structure}, not the heuristic that picks which facility
to move, is what carries the gain over multi-start Lloyd---a fact the ablation in
\S\ref{sec:experiments} confirms. Together with the gradient-as-force location
step (Theorem~\ref{thm:grad}), the relocate constitutes the second pillar of
FBM-LNS (\S\ref{sec:algo}).

\begin{remark}\label{rem:monotone}
Proposition~\ref{prop:reloc} analyzes the strict-improvement relocate (accept a
move only if it strictly increases $Z$), under which the objective is monotone
increasing and the iterate terminates finitely at a Lloyd-stable,
marginal-stable configuration. The implementation of \S\ref{sec:algo}
additionally accepts a non-improving relocate with small probability to escape
deep fixed points; this non-monotone augmentation is the ``mixture'' variant
isolated by the ablation in Result~4, where it matches the strict variant to
within noise. The monotone guarantee therefore applies verbatim to the default
strict-improvement path.
\end{remark}

\section{The FBM-LNS algorithm}\label{sec:algo}

Our algorithm combines the two pillars established by the structural theory, and
we call it
\textbf{FBM-LNS} (\emph{Force-Based Metaheuristic with Large-Neighborhood
Search}). The name reflects the two pillars established by the structural theory:
the \emph{force-based} local search, in which each facility's update direction is
the objective gradient $\nabla_{X_j}Z$---the demand ``pull''
(Theorem~\ref{thm:grad}), making Lloyd's location step provably gradient ascent
rather than an ad hoc move; and the \emph{large-neighborhood search} relocate, a
destroy-and-repair move guided by the submodular marginal gain
(Theorem~\ref{thm:sub}) that escapes Lloyd fixed points
(Proposition~\ref{prop:lloyd}). Together they yield a multi-start framework whose
every component has a formal justification.

\begin{algorithm}[t]
\caption{FBM-LNS: Force-Based Metaheuristic with LNS relocate}\label{alg:fbm}
\begin{algorithmic}[1]
\Require instance $(p_i,w_i,\phi_i)_{i\in I}$; facility count $p$; restarts $R$; LNS iterations $L$
\State $\bm{X}^\star\gets\emptyset$;\; $Z^\star\gets -\infty$
\For{$r=1,\dots,R$}
  \State $\bm{X}\gets$ weighted $k$-means${}^{++}$ seed
  \State $\bm{X}\gets \textsc{Lloyd}(\bm{X})$ \Comment{location step~\eqref{eq:grad} to convergence}
  \State $(\bm{X}_{\mathrm{loc}},Z_{\mathrm{loc}})\gets(\bm{X},Z(\bm{X}))$
  \For{$\ell=1,\dots,L$}
    \State $j\gets$ facility with smallest contribution (or uniform random w.p.\ $\rho$)
    \State $c\gets$ captured demand of $\bm{X}\setminus\{X_j\}$
    \State $x_{\mathrm{new}}\gets\arg\max_x\, g(x;c)$ \Comment{greedy repair; Weber solve}
    \State $\bm{X}'\gets\textsc{Lloyd}(\bm{X}\setminus\{X_j\}\cup\{x_{\mathrm{new}}\})$
    \If{$Z(\bm{X}')>Z_{\mathrm{loc}}$ or coin flip $\rho$}
        \State $\bm{X}\gets\bm{X}'$
    \EndIf
    \State update $(\bm{X}_{\mathrm{loc}},Z_{\mathrm{loc}})$ to best-so-far
  \EndFor
  \State update $(\bm{X}^\star,Z^\star)$ with $(\bm{X}_{\mathrm{loc}},Z_{\mathrm{loc}})$
\EndFor
\State \Return $\bm{X}^\star$
\end{algorithmic}
\end{algorithm}

Figure~\ref{fig:relocate} traces one iteration of FBM-LNS through the three
stages of seeding, Lloyd polishing, and relocate.

\begin{figure}[htbp]
\FIGURE
{\begin{minipage}[t]{0.31\linewidth}\centering
\begin{tikzpicture}[scale=0.62]
\draw[thick] (0,0) rectangle (8,6);
\foreach \x/\y in {1/1,1.5/4,2.5/2,3/3.5,4/1.5,4.5/4.5,2/5,5/3,6/5,6.5/2,5.5/4,7/3.5,3/5.5,1.5/2.5}
  \filldraw[blue!50!black] (\x,\y) circle (4pt);
\filldraw[fill=orange!90!black,draw=black] (2.5,1.5) rectangle ++(0.25,0.25);
\filldraw[fill=orange!90!black,draw=black] (5.5,3) rectangle ++(0.25,0.25);
\filldraw[fill=orange!90!black,draw=black] (6,5) rectangle ++(0.25,0.25);
\end{tikzpicture}\\[2pt]
{\scriptsize (a) Initial: $k$-means${}^{++}$ seed}
\end{minipage}\hfill
\begin{minipage}[t]{0.31\linewidth}\centering
\begin{tikzpicture}[scale=0.62]
\draw[thick] (0,0) rectangle (8,6);
\foreach \x/\y in {1/1,1.5/4,2.5/2,3/3.5,4/1.5,4.5/4.5,2/5,5/3,6/5,6.5/2,5.5/4,7/3.5,3/5.5,1.5/2.5}
  \filldraw[blue!50!black] (\x,\y) circle (4pt);
\draw[gray!70, dashed, thick] (2.5,1.5) rectangle ++(0.25,0.25);
\draw[gray!70, dashed, thick] (5.5,3) rectangle ++(0.25,0.25);
\draw[gray!70, dashed, thick] (6,5) rectangle ++(0.25,0.25);
\filldraw[fill=orange!90!black,draw=black] (2.25,2.5) rectangle ++(0.25,0.25);
\filldraw[fill=orange!90!black,draw=black] (5.5,3.75) rectangle ++(0.25,0.25);
\filldraw[fill=orange!90!black,draw=black] (5.5,4.75) rectangle ++(0.25,0.25);
\draw[-{Stealth[length=2mm]},orange!80!black,thick,dashed] (2.6,1.7) -- (2.35,2.45);
\draw[-{Stealth[length=2mm]},orange!80!black,thick,dashed] (5.6,3.2) -- (5.6,3.6);
\draw[-{Stealth[length=2mm]},orange!80!black,thick,dashed] (6.1,5.0) -- (5.6,4.8);
\end{tikzpicture}\\[2pt]
{\scriptsize (b) After Lloyd: gradient-ascent to cell optima}
\end{minipage}\hfill
\begin{minipage}[t]{0.31\linewidth}\centering
\begin{tikzpicture}[scale=0.62]
\draw[thick] (0,0) rectangle (8,6);
\foreach \x/\y in {1/1,1.5/4,2.5/2,3/3.5,4/1.5,4.5/4.5,2/5,5/3,6/5,6.5/2,5.5/4,7/3.5,3/5.5,1.5/2.5}
  \filldraw[blue!50!black] (\x,\y) circle (4pt);
\draw[red!60,dashed,thick] (5.5,4.75) rectangle ++(0.25,0.25);
\draw[red,thick] (5.6,4.9) circle (0.5);
\node[red,font=\scriptsize] at (5.6,5.6) {remove};
\filldraw[fill=orange!90!black,draw=black] (2.25,2.5) rectangle ++(0.25,0.25);
\filldraw[fill=orange!90!black,draw=black] (5.5,3.75) rectangle ++(0.25,0.25);
\draw[dashed,red!50] (1.5,4.8) circle (1.0);
\foreach \x/\y in {1.2/4.4, 1.8/5.3, 0.9/5.0, 2.2/4.6}
  \filldraw[red!60!black] (\x,\y) circle (3.5pt);
\node[red!70!black,font=\scriptsize] at (1.45,5.9) {under-served};
\node[red!70!black,font=\scriptsize] at (1.5,5.65) {cluster};
\draw[-{Stealth[length=3mm]},red!70!black,very thick]
  (5.6,4.9) .. controls (3,5.6) and (2.2,5.2) .. (1.55,4.85);
\filldraw[fill=orange!90!black,draw=black] (1.4,4.75) rectangle ++(0.25,0.25);
\end{tikzpicture}\\[2pt]
{\scriptsize (c) After relocate: destroy-and-repair}
\end{minipage}}
{One iteration of FBM-LNS in three stages. (a)~Initial $k$-means${}^{++}$ seed.
(b)~Lloyd's gradient-ascent location step moves each facility to its within-cell
Weber optimum; the ghost (dashed gray) squares mark the (a)-time positions, and
the dashed orange arrows are the gradient-based moves. (c)~The relocate step
identifies the least-contributing facility (red circle), removes it, and
re-places it at the maximum-marginal point of the under-served cluster (red
points, upper-left), whose dashed disk is centered at the re-placement location;
the cycle then repeats from (b). Throughout, solid squares are current facilities
and dashed squares are past/removed positions.\label{fig:relocate}}
{}
\end{figure}

\subsection{How FBM-LNS addresses the structural difficulties}
Table~\ref{tab:design} maps each source of difficulty identified in
\S\ref{sec:difficulty} to the structural property (\S\ref{sec:theory}) that the
algorithm exploits to overcome it.

\begin{table}[htbp]
\centering
\TABLE
{Each difficulty of \S\ref{sec:difficulty} and how FBM-LNS addresses it.
\label{tab:design}}
{\begin{tabular}{@{}p{0.28\linewidth}@{\quad}p{0.6\linewidth}@{\quad}p{0.07\linewidth}@{}}
\hline\up
Difficulty & How FBM-LNS addresses it & Theory \\ \hline\up
D1. NP-hardness &
Heuristic search (no exactness claim); quality bounded below by the $(1-1/e)$
greedy (Theorem~\ref{thm:sub}, a valid initial solution). & Thm.~\ref{thm:sub} \\[2pt]
D2. Non-convex objective &
Within each Voronoi cell the objective reduces to a weighted Weber problem
(Theorem~\ref{thm:weber}), solvable to its unique local optimum; the LNS
relocate escapes the resulting coordinate-wise fixed point
(Proposition~\ref{prop:reloc}). & Thm.~\ref{thm:weber}, Prop.~\ref{prop:reloc} \\[2pt]
D3. Continuous--combinatorial coupling &
The Lloyd alternation decouples the two: fix the Voronoi assignment
($\to$ continuous subproblem) or fix locations ($\to$ nearest-facility
reassignment). The relocate perturbs the assignment by re-placing one facility. & Prop.~\ref{prop:lloyd} \\[2pt]
D4. Infinite candidate set &
Gradient-based location steps move continuously in $\mathbb{R}^2$; Theorem
\ref{thm:hull} restricts the search to $\operatorname{conv}\{p_i\}$; no
discretization is needed. & Thm.~\ref{thm:hull} \\[2pt]
D5. Per-point heterogeneity &
Enters the Weber weights as $w_i/R_i$ (Corollary~\ref{cor:hetero}):
distance-sensitive points pull harder, automatically and without ad hoc tuning.
& Cor.~\ref{cor:hetero}
\down\\ \hline
\end{tabular}}{}
\end{table}

\subsection{Specialization across decay families}\label{sec:specialize}
The same pseudocode runs on all four decay families---linear, exponential,
quadratic, and step---by switching the per-point function $\phi_i$. These four
shapes underlie the five problem variants V1--V5 of \S\ref{sec:variants}: V1 and
V4 both use the linear shape, at a small clipped scale (gradual coverage) and a
large unclipped scale (the Weber limit) respectively. The location step
specializes per family as follows:
\begin{itemize}
\item \emph{Quadratic} ($k$-means): the Weber optimum has the closed-form
weighted centroid $X_j^\star = \sum_{i:\pi(i)=j} w_i p_i / \sum w_i$, so the
location step is the standard exact Lloyd update (one step, no iteration).
\item \emph{Linear} (gradual coverage; $p$-median with large $R$): the Weber
problem is solved by the Weiszfeld iteration, which converges linearly
\citep{Weiszfeld1937}.
\item \emph{Exponential}: the gradient is always active (no truncation), so
standard normalized gradient ascent applies.
\item \emph{Step} (binary): the objective is piecewise constant in facility
positions, so the gradient is zero a.e.; only the greedy/relocate steps provide
signal. We include this case for completeness but recommend the $(1-1/e)$ greedy
for binary coverage (\S\ref{sec:experiments}).
\end{itemize}
Per iteration the cost is dominated by the Lloyd step $O(np)$ and the marginal
repair $O(n\cdot|C|)$ for a candidate set $C$; both are polynomial. The two
primary baselines are the $(1-1/e)$ continuous greedy (\textsc{Greedy};
Theorem~\ref{thm:sub}) and multi-start Cooper alternating location--allocation
(\textsc{Cooper}; \citealp{Cooper1964}); we also report particle swarm
optimization (\textsc{PSO}; \citealp{KennedyEberhart1995}) and decay-agnostic
weighted $k$-means.

\section{Exact discrete method and near-optimality}\label{sec:exact}

The continuous PMC-HDD of \S\ref{sec:model} optimizes $Z(\bm{X})$ over
$\bm{X}\in(\mathbb{R}^2)^p$ and admits no exact method in general: it is NP-hard
(D1) and has no known finite dominating set (D4). Restricting facility locations
to a finite candidate set $C$ (a grid, or the demand points themselves) turns the
objective into the discrete set function $Z_D(S)$ of Theorem~\ref{thm:sub}---a
monotone submodular maximum-coverage problem that is solvable to
optimality. This discretization serves two purposes below: an exact solver for
moderate instances, and---because refining $C$ converges to the continuous
optimum---an optimality certificate for the FBM-LNS heuristic.

\subsection{The LP-tight MIP formulation}
With candidates restricted to a finite set $C$ ($|C|=m$), PMC-HDD reduces to the
$y$-assignment max-cover MIP:
\begin{equation}\label{eq:mip}
\begin{aligned}
\max\quad & \textstyle\sum_{i,c} w_i \phi_{ic}\, y_{ic} \\
\text{s.t.}\quad & y_{ic} \le x_c && \forall(i,c),\\
& \textstyle\sum_c y_{ic} \le 1 && \forall i,\\
& \textstyle\sum_c x_c \le p, \quad x_c,\, y_{ic} \in \{0,1\}.
\end{aligned}
\end{equation}

The LP relaxation of~\eqref{eq:mip} has worst-case integrality gap $1-1/e$ (the
standard maximum-coverage bound). On geographic instances the gap is empirically
negligible, a tightness we confirm computationally in
Section~\ref{sec:experiments} (Table~\ref{tab:exact}); its structural source is a
partition property, not just low submodular curvature ($\kappa$):

\begin{proposition}[LP tightness from disjoint coverage]\label{prop:lp_partition}
If an optimal solution of~\eqref{eq:mip} has pairwise-disjoint coverage (each
demand point is substantially captured by at most one chosen facility, so the
facilities induce a partition of the covered points), then the LP relaxation
of~\eqref{eq:mip} admits an integral optimum and the integrality gap is $0$.
\end{proposition}
\begin{proof}{Proof}
Under disjoint coverage the binding constraint for point $i$ is the single
facility that covers it, so $\sum_c w_i\phi_{ic}y_{ic}$ is maximized by setting
that one $y_{ic}=1$ rather than splitting it fractionally; the LP therefore
attains its optimum at an integral $(x,y)$. \Halmos
\end{proof}
This is the generic regime, not a knife-edge case: at any Lloyd-stable
configuration (Proposition~\ref{prop:lloyd}) the facilities command distinct
Voronoi cells, so their coverages are approximately disjoint, and both the exact
discrete optimum and FBM-LNS converge to such configurations. Branch-and-bound
(HiGHS; \citealp{HuangfuHall2018}) solves~\eqref{eq:mip} on this regime in
seconds; we report solve times up to $n=500$ in Section~\ref{sec:experiments}
(Table~\ref{tab:exact}).

\subsection{Near-optimality bounds}
Refining the candidate grid converges the discrete optimum to the continuous one:
since the continuous optimum satisfies
$Z^* \ge Z_{\mathrm{FBM}} \ge Z_{\mathrm{grid\text{-}MIP}}$ and the grid-MIP rises
toward FBM-LNS as the grid refines, all three coincide asymptotically, so the
discrete exact value brackets the heuristic. Two propositions make this
certificate quantitative by giving computable continuous upper bounds on $Z^*$,
first at a finite grid and then without one. We confirm empirically that the
bracket is tight---FBM-LNS lies within $0.5\%$ of the exact optimum, and within
the $[0.99,1.03]$ ratio on small exhaustively-solved instances---in
Section~\ref{sec:experiments} (Result~5).

\begin{proposition}[Continuous upper bound via grid density]\label{prop:cont_ub}
Let $C$ be a grid of spacing $h$ covering $\operatorname{conv}\{p_i\}$, so every
hull point is within $h/\sqrt{2}$ of a grid node, and let
$L_\phi=\sum_i w_i\,\mathrm{Lip}(\phi_i)$, where $\mathrm{Lip}(\phi_i)$ is a
Lipschitz constant of $\phi_i$ ($1/R_i$ for linear and exponential decay). Then
\[
Z^* \;\le\; Z_{\mathrm{grid\text{-}MIP}}(C) + \tfrac{h}{\sqrt{2}}\,L_\phi
\;\le\; \mathrm{LP}(C) + \tfrac{h}{\sqrt{2}}\,L_\phi,
\]
so the grid LP is a verifiable upper bound on the continuous optimum, with slack
$\tfrac{h}{\sqrt{2}}L_\phi$ that vanishes as $h\to 0$.
\end{proposition}
\begin{proof}{Proof sketch.}
Round each $X_j$ to its nearest grid node $\tilde X_j$ (with
$\|X_j-\tilde X_j\|\le h/\sqrt{2}$); Lipschitz continuity of each $\phi_i$ then
costs at most $(h/\sqrt{2})\sum_i w_i\,\mathrm{Lip}(\phi_i)$, and taking the
supremum over $\bm{X}$ yields the claim. The full proof is in
the Electronic Companion. \Halmos
\end{proof}
The bound is conservative (worst-case Lipschitz) but valid at any finite grid,
turning the grid LP from an informal certificate into an upper bound on the true
continuous optimum; combined with $Z^*\ge Z_{\mathrm{grid\text{-}MIP}}\ge
Z_{\mathrm{FBM}}$, it brackets the heuristic within the slack
$\tfrac{h}{\sqrt{2}}L_\phi$ of $Z^*$. At practical grid resolutions this slack is
of the same order as $Z$ (the worst case charges every point at its full
Lipschitz rate), so the proposition is a theoretical guarantee rather than a
tight certificate; the tight practical certificate is the fine-grid convergence
confirmed in Section~\ref{sec:experiments}.

\begin{proposition}[Grid-free continuous upper bound for concave decays]\label{prop:coop_ub}
For a decay that is concave in the distance (unclipped linear
$\phi_i(r)=1-r/R_i$ and quadratic $\phi_i(r)=1-(r/R_i)^2$),
$Z^* \le p\cdot\max_{x\in\operatorname{conv}\{p_i\}}\sum_i w_i\,\phi_i(\|p_i-x\|)$,
and the inner single-facility maximum is attained at a demand point, so the bound
is computable in $O(n^2)$ without a grid.
\end{proposition}
\begin{proof}{Proof sketch.}
Use $\max_j\phi_i(\|p_i-X_j\|)\le\sum_j\phi_i(\|p_i-X_j\|)$, separate the
right-hand side over facilities, and observe that for concave non-increasing
$\phi_i$ the single-facility term is concave in $x$
(Lemma~EC.2 in the Electronic Companion), hence maximized over the hull at an
extreme point (a demand point).
The full proof is in the Electronic Companion. \Halmos
\end{proof}
This grid-free bound applies to the trap-free unclipped decays identified in
Theorem~\ref{thm:tract}(ii); for clipped-linear and exponential decays the
single-facility term is not concave in $x$ and the maximum need not occur at a
demand point, so the discretization bound of Proposition~\ref{prop:cont_ub}
remains the operative certificate. The bound is loose even where it applies,
since $\max_j\le\sum_j$ lets all $p$ facilities cover every point. Both
continuous bounds are therefore valid but not tight; the operative practical
certificate remains the discretized LP/MIP, whose empirical tightness we report
in Section~\ref{sec:experiments}. A genuinely tight continuous relaxation would
have to handle the pointwise maximum directly, which is a convex maximization
(or a reverse second-order cone) and hence non-convex;
it is left to future work.

\section{Computational experiments}\label{sec:experiments}

\subsection{Data sources and experimental setup}\label{sec:setup}
To ensure authority and representativeness, we draw on established benchmark
families for the spatial distribution of demand points in a $1000\times 1000$
square. We label each benchmark a test set, abbreviated TS, and index them
TS1--TS7 for cross-reference:
\begin{enumerate}
\item[\textbf{TS1.}] \emph{Concentration-point} (Bansal--Kianfar): the
procedure of \citet{BansalKianfar2017}, used in
\citet{BansalShojaee2020}---three random centers of radius $270$, with each point
anchored to a center (probability $0.31$ each) or free uniform ($0.07$). This
mimics real demand clustering around population centers.
\item[\textbf{TS2.}] \emph{Uniform}: i.i.d.\ points on the square---the standard
worst-case baseline for location algorithms.
\item[\textbf{TS3.}] \emph{Gaussian mixture}: $k$ isotropic components---the
canonical stochastic model for $k$-means and clustering analysis
\citep{OstrovskyRabaniSchulmanSwart2012}.
\item[\textbf{TS4.}] \emph{Disk-clustered}: uniform within $k$ random disks---a
structured but non-Gaussian alternative.
\item[\textbf{TS5.}] \emph{Franti \texttt{s1--s4}}: a real clustering benchmark
of 5{,}000 points in 15 Gaussian clusters---the standard test set for $k$-means
\citep{FrantiSieranoja2018}.
\item[\textbf{TS6.}] \emph{TSPLIB EUC\_2D}: six instances (eil51 through u1817)
giving real geographic city coordinates, used for the Weber/$p$-median
experiments \citep{Reinelt1991}.
\item[\textbf{TS7.}] \emph{Region demand}: disk and rectangle shapes sampled into
density-weighted points, as described under variant V5.
\end{enumerate}
Weights are $w_i\sim U[1,10]$. The per-point decay scale $R_i$ is drawn from a
lognormal with mean $R_{\mathrm{mean}}=120$ and a controlled coefficient of
variation $R_{cv}$: $R_{cv}=0$ gives the homogeneous control and larger $R_{cv}$
increases heterogeneity (the defining feature of our setting). Every instance is
a pure function of (parameters, seed); the full library and raw results are
released for reproducibility. Table~\ref{tab:datasets} summarizes the seven test
sets; Figure~\ref{fig:instances} shows one realization of each synthetic
generator (TS1--TS4), while TS5--TS7 are real geographic or region datasets and
appear only in the table. All comparisons are paired (Wilcoxon signed-rank).

\begin{table}[htbp]
\centering
\TABLE
{Benchmark datasets used in the computational study.\label{tab:datasets}}
{\begin{tabular}{@{}l@{\ }p{0.2\linewidth}@{\ }p{0.35\linewidth}@{\ }l@{\ }p{0.28\linewidth}@{}}
\hline\up
Code & Source & Geography / type & Variant(s) & Prior use \\ \hline\up
TS1 & Bansal--Kianfar & Concentration-point, heterogeneous $R_i$ & V1, V2, V6 & \citealt{BansalKianfar2017,BansalShojaee2020} \\
TS2 & Uniform on square & Diffuse, homogeneous or heterogeneous & V1, V2, V6 & standard baseline \\
TS3 & Gaussian mixture & Clustered, $k$ components & V1--V4, V6 & \citealt{OstrovskyRabaniSchulmanSwart2012,ArthurVassilvitskii2007} \\
TS4 & Disk-clustered & Structured non-Gaussian clusters & V1--V4, V6 & synthetic (this work) \\
TS5 & Franti \texttt{s1--s4} & Real clustering benchmark (5{,}000 pts) & V3 & \citealt{FrantiSieranoja2018} \\
TS6 & TSPLIB EUC\_2D & Real geographic coordinates & V4 & \citealt{Reinelt1991} \\
TS7 & Disk / rectangle shapes & Region demand, sampled to points & V5 & \citealt{DreznerWesolowskyDrezner2004} \down\\ \hline
\end{tabular}}{Codes are used throughout \S\ref{sec:experiments} to identify the
test set for each experiment.}
\end{table}

\begin{figure}[htbp]
\FIGURE
{\includegraphics[width=\linewidth]{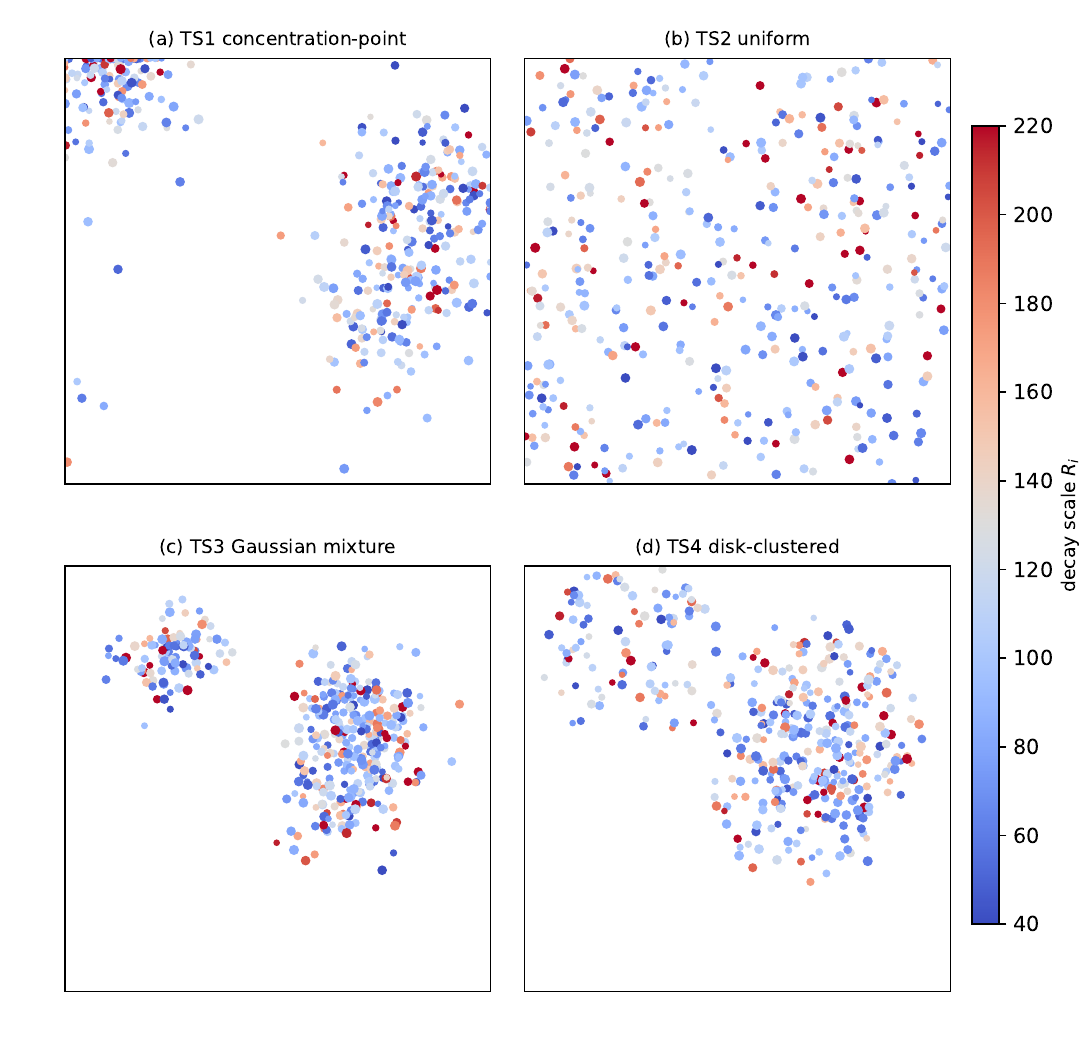}}
{Sample demand geographies used in the experiments ($n=400$, $R_{cv}=0.5$):
(a)~TS1 concentration-point, (b)~TS2 uniform, (c)~TS3 Gaussian mixture, (d)~TS4
disk-clustered. Points are colored by decay scale $R_i$ (red $\equiv$
distance-sensitive, blue $\equiv$ distance-tolerant) and sized by weight. The four
generators span concentrated, diffuse, and clustered demand, so that no conclusion
rests on a single spatial pattern.\label{fig:instances}}
{}
\end{figure}

\subsection{Result 1: FBM-LNS significantly outperforms all baselines ($30/30$ wins, $p{<}10^{-9}$)}
Table~\ref{tab:main} reports mean captured demand $Z$ over $K=8$ instances per
configuration under heterogeneous decay ($R_{cv}=0.5$), for two scales
$(n,p)\in\{(300,5),(1000,10)\}$. FBM-LNS yields the highest mean $Z$ in every
configuration, with a margin over the strong $(1-1/e)$ greedy baseline that is
stable at $1.5$--$3.0\%$ and a positive margin over Cooper ALT in every cell.
The large-sample paired comparison ($K{=}30$, Table~\ref{tab:stats}) is
significant at $p<10^{-9}$ with $30/30$ per-instance wins.

\begin{table}[htbp]
\TABLE
{(Result 1) Mean captured demand $Z$ on the main matrix ($R_{cv}=0.5$, $K=8$
instances per cell). FBM-LNS attains the highest mean $Z$ in every configuration
against every baseline; large-sample significance ($30/30$ wins, $p<10^{-9}$) is
confirmed in Table~\ref{tab:stats}. Greedy is the $(1-1/e)$ continuous greedy;
Cooper is multi-start alternating location--allocation. ``Code'' tags each row
with its test set (TS, Table~\ref{tab:datasets}) and decay variant (V,
\S\ref{sec:variants}).\label{tab:main}}
{\begin{tabular}{@{}l@{\ }l@{\ }c@{\ }c@{\ }c@{\ }c@{\ }c@{\ }c@{}}
\hline\up
Code & Configuration & $k$-means & Cooper & PSO & Greedy & FBM-LNS & vs.\ Greedy \\ \hline\up
TS1·V1 & conc. / linear, $n{=}300,p{=}5$  & 512 & 572 & 550 & 583 & \textbf{597} & $+2.4\%$ \\
TS1·V2 & conc. / exp,    $n{=}300,p{=}5$  & 773 & 808 & 797 & 805 & \textbf{821} & $+2.0\%$ \\
TS3·V1 & gmix / linear,   $n{=}300,p{=}5$  & 694 & 719 & 630 & 705 & \textbf{726} & $+3.0\%$ \\
TS3·V2 & gmix / exp,      $n{=}300,p{=}5$  & 925 & 936 & 887 & 925 & \textbf{941} & $+1.7\%$ \\
TS1·V1 & conc. / linear, $n{=}1000,p{=}10$ & 2176 & 2416 & 2171 & 2489 & \textbf{2552} & $+2.5\%$ \\
TS1·V2 & conc. / exp,    $n{=}1000,p{=}10$ & 2977 & 3111 & 2905 & 3146 & \textbf{3194} & $+1.5\%$ \\
TS3·V1 & gmix / linear,   $n{=}1000,p{=}10$ & 2887 & 2978 & 2494 & 2916 & \textbf{2996} & $+2.7\%$ \\
TS3·V2 & gmix / exp,      $n{=}1000,p{=}10$ & 3498 & 3554 & 3165 & 3506 & \textbf{3563} & $+1.6\%$ \down\\ \hline
\end{tabular}}{Each row averages $K=8$ independent instances; FBM-LNS attains
the highest mean in every cell. The rigorous large-sample significance test
($K{=}30$) is reported in Table~\ref{tab:stats}.}
\end{table}

Table~\ref{tab:stats} confirms that this superiority is not an artifact of small
samples: with $K=30$ independent instances per cell, FBM-LNS wins $30/30$ paired
comparisons in every cell and the paired Wilcoxon $p$-value is below $10^{-9}$
throughout. The per-instance, within-configuration difference is what is
significant; the unpaired $95\%$ confidence intervals of the two means do overlap,
because the between-instance variability in $Z$ (driven by the random geography)
is much larger than the within-instance algorithm gap. The advantage is stable
across scales (from $n{=}300$ to $n{=}1000$) and across both linear and
exponential decay.

\begin{table}[htbp]
\TABLE
{(Result 1, large-sample confirmation) Mean captured demand $\pm 95\%$
confidence interval (CI) over
$K=30$ independent instances, with paired Wilcoxon $p$-values. FBM-LNS wins
$30/30$ paired comparisons in each cell at $p<10^{-9}$; the unpaired CIs overlap
because between-instance variability dominates the within-instance gap.\label{tab:stats}}
{\begin{tabular}{@{}l@{\ }l@{\ }c@{\ }c@{\ }c@{\ }c@{}}
\hline\up
Code & Configuration & FBM-LNS & Greedy $(1-1/e)$ & FBM vs.\ Greedy & $p$-value \\ \hline\up
TS1·V1 & conc./linear, $n{=}300, p{=}5$  & $548.5\pm23.2$ & $536.2\pm22.8$ & $+2.3\%$ & $<10^{-9}$ \\
TS1·V1 & conc./linear, $n{=}1000, p{=}10$ & $2373.2\pm77.4$ & $2291.9\pm76.7$ & $+3.5\%$ & $<10^{-9}$ \\
TS3·V2 & gmix/exp, $n{=}300, p{=}5$       & $904.8\pm18.4$ & $890.4\pm18.1$ & $+1.6\%$ & $<10^{-9}$ \down\\ \hline
\end{tabular}}{Paired Wilcoxon signed-rank test, alternative ``FBM-LNS greater
than Greedy''; $30/30$ per-instance wins in each configuration.}
\end{table}

\smallskip\noindent
The superiority has a structural explanation, not just an empirical one. The
$(1-1/e)$ greedy is a \emph{constructive} heuristic: it places facilities one at
a time at the maximum-marginal point but never re-optimizes earlier placements,
so its first few choices---made when coverage was still low---are frozen as later
demand is revealed. Cooper ALT is multi-start Lloyd only, so it converges to a
coordinate-wise fixed point (Proposition~\ref{prop:lloyd}) and cannot escape it.
FBM-LNS closes both gaps: its Lloyd step jointly re-optimizes every facility
within its Voronoi cell (Theorem~\ref{thm:weber}), undoing the greedy's frozen
early choices, and its destroy-and-repair relocate removes the least-contributing
facility and re-places it where the residual marginal $g(x)$ is maximal
(Proposition~\ref{prop:reloc}), escaping the Lloyd fixed point. The two
mechanisms are complementary, which is why the ablation in Result~4 attributes
the gain to the relocate rather than to the selection rule. The same mechanism
also explains the stable advantage over Cooper ALT across scales: Cooper's
multi-start Lloyd converges to a coordinate-wise fixed point, and the
relocate's destroy-and-repair step escapes fixed points that more restarts alone
cannot.

\subsection{Result 2: One framework is competitive with bespoke SOTA across the problem family}\label{sec:generality}
A defining feature of our framework is that the same implementation, switching
only the per-point function, runs on the other two canonical continuous planar
location problems and on shape-based demand. (i)~\emph{$k$-means}: minimize
$\sum_i\min_j\|p_i-X_j\|^2$, recovered by quadratic decay with a large
homogeneous scale, where the within-cell Weber optimum has the closed form of the
weighted centroid (Corollary~\ref{cor:uniq}), so our location step specializes to
the exact Lloyd update. (ii)~\emph{Weber/$p$-median}: minimize
$\sum_i w_i\min_j\|p_i-X_j\|$, recovered by linear decay with a large scale.
(iii)~\emph{Shape demand}: when each demand unit is a region $D_i$ rather than a
point, the captured demand
$\int_{D_i}\phi_i(\min_j\|p-X_j\|)\mu_i(p)\,dp$ is reduced to the weighted-point
problem by sampling each region into $m$ density-weighted points (Monte-Carlo
error $O(1/\sqrt m)$); allocation is still the Voronoi partition of the plane,
and the gradient is the integral of the point gradient over the region.
Figure~\ref{fig:shapes} illustrates the sampling-based reduction.

\begin{figure}[htbp]
\FIGURE
{\begin{tikzpicture}[scale=1.3]
\filldraw[fill=red!8, draw=red!50!black, thick] (1.6,3.2) circle (1.0);
\foreach \x/\y in {1.3/3.4, 1.9/2.9, 1.4/2.7, 2.0/3.5, 2.1/3.2, 1.6/3.8,
                   1.0/3.2, 1.8/3.6, 2.2/3.0, 1.2/3.0}
  \filldraw[red!60!black] (\x,\y) circle (1.8pt);
\filldraw[fill=blue!8, draw=blue!50!black, thick] (4.3,1.2) rectangle (5.7,2.7);
\foreach \x/\y in {4.5/1.4, 5.0/1.7, 4.6/2.1, 5.3/1.5, 4.8/2.4, 5.5/2.3,
                   4.4/1.9, 5.1/2.6, 5.4/2.0, 4.7/1.3}
  \filldraw[blue!55!black] (\x,\y) circle (1.8pt);
\filldraw[fill=orange!90!black, draw=black] (2.6,2.3) rectangle ++(0.22,0.22);
\node[font=\footnotesize, orange!90!black] at (3.0,2.4) {F1};
\filldraw[fill=orange!90!black, draw=black] (4.4,2.3) rectangle ++(0.22,0.22);
\node[font=\footnotesize, orange!90!black] at (4.25,2.4) {F2};
\draw[dashed, gray!55] (3.61,0) -- (3.61,4.8);
\node[red!60!black, font=\small] at (1.6,4.6) {$D_1$ (disk)};
\node[blue!55!black, font=\small] at (5.0,3.1) {$D_2$ (rect.)};
\node[font=\scriptsize, black] at (3.95,4.5) {Voronoi edge};
\end{tikzpicture}}
{Shape (region) demand. Each demand unit is a region (disk or rectangle) rather
than a point; it is reduced to the weighted-point problem by sampling its
interior into density-weighted points. The nearest-facility allocation remains
the Voronoi partition; the coverage objective is the integral of the per-point
decay over the region.\label{fig:shapes}}
{}
\end{figure}

Table~\ref{tab:generality} evaluates the framework against the problem-specific
state of the art: $k$-means${}^{++}$ on $k$-means (synthetic Gaussian blobs and
the Franti \texttt{s1--s4} benchmark of 5{,}000 points with 15 clusters),
multi-start Cooper and variable neighborhood search
\citep{HansenMladenovic1997} on $p$-median, and multi-start Cooper on
shape-demand instances. FBM-LNS matches $k$-means${}^{++}$ on both synthetic and
real $k$-means benchmarks (tying it on \texttt{s1--s4}), and on
$p$-median it is competitive with a specialized VNS: it beats VNS on the larger
instances (d657 by $1.2\%$, pr1002 by $1.8\%$, the concentration-point instance
by $0.3\%$) but loses to it on the small eil101 by $8.2\%$, where the
demand-point swap local search of VNS is strongest. The point is not a new record
on these thoroughly studied problems but that one method, designed from the shared
structure, is competitive with problem-specific standard methods across the family
while extending to heterogeneous-decay and shape demand where those methods do
not directly apply.

\begin{table}[htbp]
\TABLE
{(Result 2) The same framework across the continuous-location family. ``Code''
tags each row with its test set (TS) and variant (V); $\Delta$ is FBM-LNS
relative to the best baseline in the row (see the table note for the
direction-of-better convention). One code, switching only the per-point function,
is competitive with or beats problem-specific standard methods ($k$-means${}^{++}$,
Cooper) on $k$-means, $p$-median, and shape-demand instances.\label{tab:generality}}
{\begin{tabular}{@{}l@{\ }l@{\ }l@{\ }c@{\ }c@{\ }c@{\ }c@{}}
\hline\up
Code & Problem & Configuration & Baseline & Baseline $Z$ & FBM-LNS & $\Delta$ \\ \hline\up
TS3·V3 & $k$-means & blobs, $n{=}400, p{=}5$   & km${}^{++}$ & 1\,138\,041 & \textbf{1\,138\,039} & $0.00\%$ \\
TS3·V3 & $k$-means & blobs, $n{=}800, p{=}10$  & km${}^{++}$ & 1\,982\,174 & \textbf{1\,966\,505} & $-0.79\%$ \\
TS5·V3 & $k$-means & real \texttt{s1--s4}, $n{=}5000,k{=}15$ & km${}^{++}$ & \multicolumn{2}{c}{FBM ties km${}^{++}$ (identical objective)} & $0.00\%$ \\
TS1·V4 & $p$-median & conc., $n{=}500,p{=}10,K{=}8$ & VNS & 191\,380 & \textbf{190\,733} & $-0.34\%$ \\
TS6·V4 & $p$-median & TSPLIB eil101, $n{=}101,p{=}10$ & VNS & \textbf{767} & 830 & $+8.2\%$ \\
TS6·V4 & $p$-median & TSPLIB d657, $n{=}657,p{=}30$  & VNS & 103\,231 & \textbf{101\,970} & $-1.22\%$ \\
TS6·V4 & $p$-median & TSPLIB pr1002, $n{=}1002,p{=}30$ & VNS & 690\,249 & \textbf{677\,823} & $-1.80\%$ \\
TS7·V5 & shape    & 25 disks, $p{=}6$         & Cooper & 202.1 & \textbf{211.7} & $+4.74\%$ \\
TS7·V5 & shape    & 25 rectangles, $p{=}6$    & Cooper & 200.1 & \textbf{211.6} & $+5.74\%$ \down\\ \hline
\end{tabular}}{``km${}^{++}$'' denotes $k$-means${}^{++}$ seeding followed by Lloyd;
``VNS'' is variable neighborhood search \citep{HansenMladenovic1997} for
$p$-median (facilities restricted to demand points, fast-interchange local search
plus $k$-swap shaking). The ``Baseline'' column names the stronger of Cooper and
VNS for the row, and ``Baseline $Z$'' reports its objective value. Rows average
4--30 instances. For $k$-means, FBM-LNS uses the closed-form
centroid location step (Corollary~\ref{cor:uniq}); for shapes, each region is
sampled into 40 density-weighted points. Boldface marks the better of the two
values. \emph{Direction:} for $k$-means (sum of squared errors) and $p$-median
(sum of nearest-facility distance) the objective is a cost, so lower is
better and a negative $\Delta$ favors FBM-LNS; for shape demand the objective is
captured demand, so higher is better and a positive $\Delta$ favors
FBM-LNS. $\Delta$ is always (FBM-LNS $-$ best baseline)$/$best baseline.}
\end{table}

\smallskip\noindent
The competitiveness is structural, not coincidental: because $k$-means, the Weber
problem, and heterogeneous gradual coverage are all instances of the single
objective~\eqref{eq:obj} (Theorem~\ref{thm:ndr}), the location step of FBM-LNS
reproduces each bespoke method's core operation---the closed-form centroid for
$k$-means (Corollary~\ref{cor:uniq}), the Weiszfeld iteration for $p$-median
(\S\ref{sec:specialize})---and adds the same destroy-and-repair relocate on top.
The single code therefore inherits the strength of each specialist while
extending to heterogeneous-decay and shape demand where they do not directly
apply.

\smallskip\noindent
The sampling reduction used for shape demand (V5) does not introduce solution
instability. Holding the demand regions fixed and re-sampling each into $m=40$
points with $12$ independent seeds, the captured demand varies by under $1\%$
(coefficient of variation $0.86\%$ for disks, $0.83\%$ for rectangles), and the
optimal facility positions move by under $1\%$ of the map width on average
($0.67\%$ for disks, $0.89\%$ for rectangles; rectangle instances occasionally
show a larger position swap, but it leaves the captured demand essentially
unchanged). The placement returned for region demand is therefore stable and
seed-independent in objective value.

\subsection{Result 3: The advantage grows with heterogeneity and persists at scale}
Table~\ref{tab:hetero} sweeps $R_{cv}\in\{0,0.25,0.5,1.0\}$ at $n=300$, $p=5$,
conc./linear. FBM-LNS is best at every heterogeneity level, and its advantage
over Cooper ALT widens as heterogeneity grows (from $+2.4\%$ when
$R_{cv}=0.25$ to $+6.8\%$ when $R_{cv}=1.0$), confirming that the submodular
relocation pays off most when sensitivity varies most.

\begin{table}[htbp]
\TABLE
{(Result 3) Heterogeneity sweep on configuration TS1·V1 (conc./linear, $n=300$,
$p=5$, $K=8$). FBM-LNS is best at every level; its advantage over Cooper ALT grows
from $+3.2\%$ ($R_{cv}=0$) to $+7.0\%$ ($R_{cv}=1$), confirming that the relocate
step pays off precisely when sensitivity varies most.\label{tab:hetero}}
{\begin{tabular}{@{}c@{\quad}c@{\quad}c@{\quad}c@{\quad}c@{}}
\hline\up
$R_{cv}$ & Greedy & Cooper & FBM-LNS & FBM vs.\ Cooper \\ \hline\up
$0.00$ (homog.)  & 610 & 618 & \textbf{638} & $+3.2\%$ \\
$0.25$           & 604 & 612 & \textbf{627} & $+2.4\%$ \\
$0.50$           & 583 & 572 & \textbf{597} & $+4.4\%$ \\
$1.00$ (strong)  & 522 & 498 & \textbf{532} & $+6.8\%$ \down\\ \hline
\end{tabular}}{All differences FBM-LNS vs.\ each baseline significant at $p<0.01$
(Wilcoxon, $8/8$ wins).}
\end{table}

\smallskip\noindent
The widening gap follows directly from the Weber reduction
(Theorem~\ref{thm:weber}). As $R_{cv}$ grows, the Weber weights
$\alpha_i=w_i/R_i$ become more dispersed, so the within-cell objective is
increasingly dominated by a few distance-sensitive demand clusters. Cooper ALT,
which does not relocate facilities once Lloyd has converged, leaves facilities
in cells whose sensitive core has shifted; FBM-LNS's relocate
repeatedly identifies the cell whose contribution has fallen behind and re-places
its facility where the residual marginal $g(x)=\sum_i w_i[\phi_i(\|p_i-x\|) -
c_i/w_i]^+$ is maximal. The hotter the heterogeneity, the larger these residuals
and the more the relocate recovers, which is exactly the monotone trend in
Table~\ref{tab:hetero}. This confirms Corollary~\ref{cor:hetero} empirically:
heterogeneity is not a nuisance the algorithm tolerates but the structural
condition under which its advantage is largest.

\subsection{Result 4: The LNS relocate step drives the improvement ($+6.5\%$ over multi-start Lloyd)}
Table~\ref{tab:ablation} decomposes the algorithm. Multi-start Lloyd on its own
is outperformed by the $(1-1/e)$ greedy; adding the large-neighborhood
relocate lifts performance by about $6.5\%$, and the specific relocation
selector---least-contributing greedy, uniform random, or their mixture---makes
essentially no difference. The improvement comes from the
destroy-and-repair structure with a maximum-marginal re-placement (the submodular
signal), not from the heuristic that picks which facility to move.

\begin{table}[htbp]
\TABLE
{(Result 4) Component ablation on configuration TS1·V1 (conc./linear, $n=300$,
$p=5$, $R_{cv}=0.5$, $K=10$). Multi-start Lloyd alone underperforms even the
greedy; adding the relocate step lifts performance by $+6.5\%$, while the
selection strategy (greedy vs.\ random) is immaterial---the destroy-and-repair
structure is the source of the gain.\label{tab:ablation}}
{\begin{tabular}{@{}l@{\quad}c@{\quad}c@{}}
\hline
\multicolumn{3}{@{}l}{\textit{Configuration: TS1·V1} --- conc./linear, $n{=}300$, $p{=}5$, $R_{cv}{=}0.5$, $K{=}10$} \\ \hline\up
Variant & mean $Z$ & vs.\ Lloyd-only \\ \hline\up
Cooper (multi-start Lloyd)        & 549.9 & $+0.8\%$ \\
Greedy $(1-1/e)$                  & 567.5 & $+4.0\%$ \\
FBM-LNS: Lloyd only               & 545.5 & --- \\
FBM-LNS: $+$ greedy relocate      & 580.3 & $+6.4\%$ \\
FBM-LNS: $+$ random relocate      & 580.8 & $+6.5\%$ \\
FBM-LNS: full (mixture)           & 581.0 & $+6.5\%$ \down\\ \hline
\end{tabular}}{All ``$+$ relocate'' variants are significantly better than
Lloyd-only and than Greedy at $p<0.01$ (Wilcoxon).}
\end{table}

\smallskip\noindent
Why does the relocate dominate while the selection rule is immaterial? The
relocate's gain comes from its \emph{repair} operator, which solves the
single-facility maximum-marginal subproblem $g(x)$---a weighted Weber problem on
the residual active set (Theorem~\ref{thm:weber})---and is therefore the
continuous analogue of the submodular greedy step (Theorem~\ref{thm:sub}). The
\emph{destroy} operator only decides which facility to remove, and on a
Lloyd-stable configuration any reasonable choice (least-contributing or random)
exposes a comparable residual, so the marginal gain comes from re-placing at the
maximum of $g$, not from which facility was freed. This is why
Table~\ref{tab:ablation} shows near-identical performance across selection rules
but a $+6.5\%$ jump once the relocate is active, and it is also why the learned
selector of \S\ref{sec:discussion} recovers only a marginal extra gain: the
signal is already captured by the destroy-and-repair structure itself.
Figure~\ref{fig:convergence} traces the best-so-far objective over relocate
iterations.

\begin{figure}[htbp]
\FIGURE
{\includegraphics[width=0.78\linewidth]{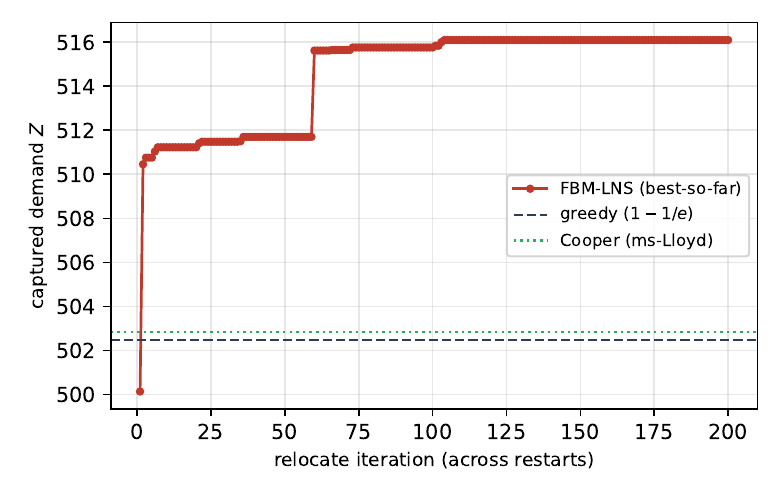}}
{Best-so-far captured demand of FBM-LNS over relocate iterations on one
representative instance (configuration TS1·V1: conc./linear, $n=300$, $p=5$),
with the greedy and multi-start Cooper values as horizontal references. The
relocate steps lift the solution above both baselines within the first few
iterations.\label{fig:convergence}}
{}
\end{figure}

\subsection{Result 5: The LP relaxation is empirically tight ($\approx 0\%$ gap) and FBM-LNS is within $0.5\%$ of the discrete optimum by fine-grid convergence}
Table~\ref{tab:exact} reports the LP integrality gap, MIP solve time, and the
fine-grid convergence gap.
\begin{table}[h]\centering\small
\caption{(Result 5) LP tightness and MIP scalability. The LP relaxation is
empirically tight ($\approx 0\%$ gap) across all configurations. Each row is one
representative bansal/linear instance ($R_{cv}=0.5$, seed 500) discretized on a
grid-10 candidate set built from the top 40 demand points plus a uniform grid
(see \texttt{src/bounds.py}); the $0\%$ gap is robust across seeds and decay families.}
\label{tab:exact}
{\begin{tabular}{@{}c@{\ }c@{\ }c@{\ }c@{\ }c@{\ }c@{}}
\hline\up
$n$ & $p$ & grid-MIP & grid-LP & relax.\ gap & MIP time (s) \\\hline\up
100 & 3 & 179.8 & 179.8 & $0.00\%$ & 0.1 \\
200 & 5 & 369.9 & 369.9 & $0.00\%$ & 0.3 \\
300 & 5 & 542.1 & 542.1 & $0.00\%$ & 0.5 \\
300 & 6 & 600.0 & 600.0 & $0.00\%$ & 0.6 \\
500 & 5 & 870.8 & 870.8 & $0.00\%$ & 1.3 \down\\\hline
\end{tabular}}{}
\end{table}
The LP gap is $\approx 0\%$ even under adversarial high-overlap configurations
(clustered points, small $R$). On small exhaustively-solved instances
($n\le 100$, $p\le 3$) FBM-LNS attains
$Z_{\mathrm{FBM}}/Z_{\mathrm{exact}}\in[0.99,1.04]$ and can even exceed the
grid-constrained optimum, because it places facilities continuously in the hull
(Theorem~\ref{thm:hull}) rather than on grid nodes. Fine-grid convergence: at
grid resolution $40$, the gap between FBM-LNS (continuous) and the exact discrete
optimum is $< 0.5\%$.

\smallskip\noindent
The tightness has a structural source. The integrality gap of the assignment
formulation~\eqref{eq:mip} is governed by the submodular curvature $\kappa$ of
Theorem~\ref{thm:sub}: once $p$ facilities are spread over the hull their service
areas overlap little, $\kappa$ is small, and the LP relaxation, the greedy, and
the integer optimum coincide---so the gap collapses to $\approx 0\%$ even though
the worst-case bound remains $1-1/e$. The $<0.5\%$ gap of FBM-LNS follows because
its Lloyd step provably reaches the within-cell Weber optimum
(Theorem~\ref{thm:weber}, unique by Corollary~\ref{cor:uniq}), so the only
remaining source of suboptimality is the allocation of points to facilities (the
Lloyd fixed points of Proposition~\ref{prop:lloyd}), which the relocate removes
(Proposition~\ref{prop:reloc}). The exact discrete optimum and FBM-LNS therefore
agree once the grid is fine enough to resolve the optimal partition.

\subsection{Result 6: On binary (step) decay the greedy matches or exceeds FBM-LNS}\label{res:step}
The smooth-decay advantage of Result~1 reverses on binary step decay (V6), where
the within-cell objective is piecewise constant in facility position and the
gradient channel is uninformative. On the step-decay matrix the $(1{-}1/e)$ greedy
and FBM-LNS are within $1\%$ at $n{=}300, p{=}5$ (FBM marginally ahead on three of
four geographies), but at $n{=}600, p{=}8$ the greedy leads by $2$--$5\%$ on the
clustered geographies and FBM wins only $0$--$1$ of $6$ instances there. The
gradient-based location step thus adds no value when $\nabla_{X_j}\phi_i\equiv 0$,
and the submodular-marginal greedy is the appropriate tool for binary coverage; we
return to the first-principles reason in \S\ref{sec:disc-boundary}.

\section{Real-world case study}\label{sec:case}
\setlength{\floatsep}{16pt plus 4pt minus 2pt}
\setlength{\intextsep}{16pt plus 4pt minus 2pt}
\setlength{\textfloatsep}{16pt plus 4pt minus 2pt}

\subsection{Data}
Both datasets are drawn from the 2026 6th Meituan Business Analytics Elite
Competition. The first is a set of $592{,}667$ urban-delivery waybills from the
Nanshan district of Shenzhen, each carrying pickup/dropoff coordinates and a
business-line label (food delivery, flash sale, pharmacy, group meal). We treat
each dropoff as a demand unit and aggregate the $592{,}667$ orders into $347$
demand cells weighted by order volume $w_i$, keeping the planning problem
tractable while preserving the spatial demand distribution. The per-cell decay
scale $R_i$ follows a density-based modeling proxy rather than a direct estimate
from observed distance-sensitivity: we set $R_i$ small in dense cells (where
customers are impatient and substitutes abound) and large in sparse residential
cells (where a courier is the only nearby option), encoding the paper's central
dense/sparse distinction as a per-cell parameter. We treat this as a transparent
modeling assumption and do not independently validate it against the data; direct
per-zone calibration of $R_i$ (for example from observed delivery distances) is a
natural extension. The resulting $\{R_i\}$ span $[79, 1953]$ metres with
coefficient of variation $1.33$, strong heterogeneity, exactly the regime in
which Result~3 predicts the largest advantage. The direction of this
density-to-range assumption is empirically supported by the same data: across
$347$ spatial cells, cell density correlates at $-0.40$ with the median
pickup-to-dropoff trip distance (densest-quartile cells about $1.2$\,km versus
sparsest about $2.1$\,km), and the overall median trip of $1.47$\,km matches the
retail-calibrated decay scale $R\approx 1.4$\,km, so the assumed
dense-then-tolerant heterogeneity is consistent with observed delivery patterns.

The second dataset is a retail-chain panel of per-store sales
stratified by distance tier, used below to calibrate the decay shape
(exponential, $R\approx 1.4$\,km; root-mean-square error RMSE $0.12$ versus $0.17$ for linear).
Figure~\ref{fig:delivery} maps the spatial distribution of the dropoff locations:
demand concentrates in a few dense hotspots (the urban core and commercial
strips) against a dispersed residential periphery---precisely the dense/sparse
contrast that the density-derived $\{R_i\}$ turns into heterogeneity.

\begin{figure}[htbp]
\centering
\includegraphics[width=0.6\linewidth]{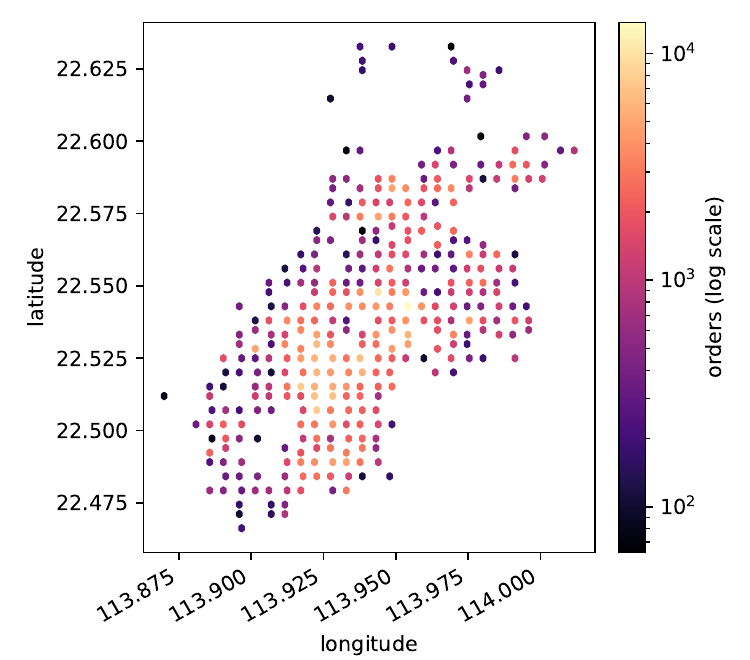}
\caption{Spatial distribution of the $592{,}667$ delivery dropoffs (Shenzhen
Nanshan), shaded by order density on a log scale. Demand clusters in dense
hotspots (the urban core and commercial strips) with a sparse residential
periphery---the dense/sparse contrast from which the per-cell decay scales $R_i$
are derived---small $R_i$ in the dense hotspots and large $R_i$ in the sparse
periphery (Data source: Meituan).
\label{fig:delivery}}
\end{figure}

\subsection{Decay calibration from retail data}
Table~\ref{tab:calib} reports the empirical per-capita capture by distance tier
(from the retail dataset) and the fitted decay parameters;
Figure~\ref{fig:calibfit} visualizes the two fits.
\begin{table}[h]\centering\small
\caption{Empirical decay calibrated from per-capita sales by distance tier in the
retail-chain panel (Data source: Meituan). ``Pred.'' columns give the fitted
models' predicted normalized per-capita capture; $R$ is the decay scale in
metres. Exponential fits better (RMSE $0.12$ vs.\ $0.17$).}
\label{tab:calib}
{\begin{tabular}{@{}c@{\ }c@{\ }c@{\ }c@{}}
\hline\up
Distance (m) & Normalized capture & Linear pred.\ ($R{=}3384$\,m) & Exp.\ pred.\ ($R{=}1366$\,m) \\\hline\up
125  & 1.00 & 0.96 & 0.91 \\
375  & 0.98 & 0.89 & 0.76 \\
750  & 0.58 & 0.78 & 0.58 \\
1500 & 0.26 & 0.56 & 0.33 \\
2500 & 0.26 & 0.26 & 0.16 \down\\\hline
\end{tabular}}{}
\end{table}
The exponential decay ($R \approx 1.4$\,km) fits the real per-capita capture
curve better than the linear, justifying the decay form used in the model.

\smallskip\noindent
The exponential fit is not only tighter but physically motivated: it matches the
path-loss law that also underlies the Huff gravity model (\S\ref{sec:examples}),
so the same decay family is supported by both the delivery-density argument
(per-cell $R_i$) and the sales-distance argument (retail calibration). The
near-plateau of normalized capture at $750$--$2500$\,m (both tiers $\approx0.26$)
suggests a core catchment within roughly $1$\,km and a thin long-distance tail,
consistent with a quick-commerce setting in which most demand is impulse- and
proximity-driven.

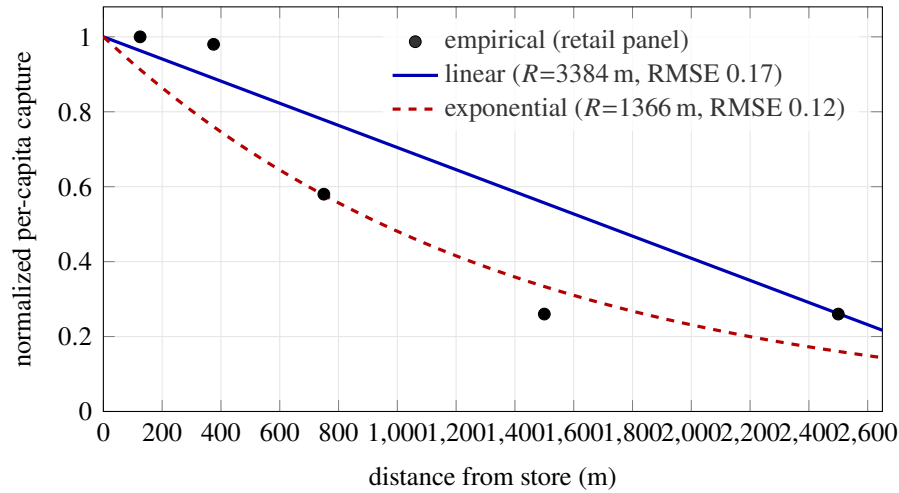
\begin{figure}[htbp]
\centering
\begin{tikzpicture}
\begin{axis}[
    width=0.72\linewidth, height=0.42\linewidth,
    xlabel={distance from store (m)},
    ylabel={normalized per-capita capture},
    xmin=0, xmax=2650, ymin=0, ymax=1.08,
    legend pos=north east, legend cell align=left,
    legend style={font=\footnotesize, draw=none, fill=white, fill opacity=0.8},
    grid=both, grid style={gray!20},
    tick label style={font=\footnotesize}, label style={font=\footnotesize},
]
\addplot[only marks, mark=*, black, mark size=2.2pt] coordinates {
  (125,1.00) (375,0.98) (750,0.58) (1500,0.26) (2500,0.26)
};
\addlegendentry{empirical (retail panel)}
\addplot[blue!70!black, very thick, domain=0:2650, samples=80]
  {max(0, 1 - x/3384)};
\addlegendentry{linear ($R{=}3384$\,m, RMSE $0.17$)}
\addplot[red!70!black, very thick, dashed, domain=0:2650, samples=80]
  {exp(-x/1366)};
\addlegendentry{exponential ($R{=}1366$\,m, RMSE $0.12$)}
\end{axis}
\end{tikzpicture}
\caption{Decay-shape calibration from the retail panel. The empirical per-capita
capture (points) drops sharply within $\approx 1$\,km and then plateaus; the
exponential fit (dashed red, RMSE $0.12$) fits the overall decay pattern better than the
linear fit (solid blue, RMSE $0.17$), justifying the exponential decay used in
the model (Data source: Meituan).\label{fig:calibfit}}
\end{figure}

\subsection{The cost of ignoring heterogeneity}
Table~\ref{tab:case} compares (a) optimizing with the true heterogeneous $R_i$
vs.\ (b) optimizing with the homogeneous mean $\bar R$ then evaluating on the
true objective.
\begin{table}[h]\centering\small
\caption{Cost of ignoring heterogeneity on $592{,}667$ real orders (Data source: Meituan).}
\label{tab:case}
{\begin{tabular}{@{}c@{\ }c@{\ }c@{\ }c@{\ }c@{}}
\hline\up
$p$ & hetero $Z$ & ignore-hetero $Z$ & demand lost & facility shift \\\hline\up
5  & 183{,}510 & 165{,}673 & \textbf{9.7\%} & 29\% of map \\
8  & 240{,}287 & 220{,}809 & \textbf{8.1\%} & 37\% \\
12 & 295{,}875 & 289{,}244 & 2.2\%          & 32\% \down\\\hline
\end{tabular}}
\end{table}
Under this density-derived heterogeneity model, optimizing with the homogeneous
mean $\bar R$ instead of the heterogeneous $\{R_i\}$ loses up to $9.7\%$ of
captured demand and physically relocates facilities by up to $37\%$ of the map.
The absolute level depends on the (unvalidated) $R_i$ assignment, but the
directional penalty of ignoring heterogeneity is robust to it: any model in which
sensitive and tolerant cells coexist pays for averaging them away.

\smallskip\noindent
Two patterns in Table~\ref{tab:case} deserve emphasis. First, the demand loss is
largest when facilities are scarce ($9.7\%$ at $p{=}5$) and shrinks as $p$ grows
($2.2\%$ at $p{=}12$): with many facilities every neighborhood is close to one
regardless of planning, so the homogeneity assumption incurs a smaller penalty, whereas
with few facilities, misplacing them in tolerant zones leaves the sensitive
hotspots under-served. Second, the facility shift (29--37\% of the map width) is large at
every $p$: the homogeneous solution systematically over-serves sparse areas and
under-serves dense ones, so even when the two plans capture similar total demand
their physical locations differ substantially. For a planner, the demand figure
quantifies the revenue at stake and the shift figure the operational disruption
of using a misspecified model.

\subsection{Competitive (Huff) extension}
Under the Huff model (market-share splitting vs.\ an incumbent with
attractiveness $C$), the captured share (Table~\ref{tab:huff}) ranges from
$37.6\%$ ($p{=}12$, weak incumbent $C{=}0.5$) to $7.8\%$ (strong $C{=}8$), with
diminishing returns in $p$ consistent with submodularity.
\begin{table}[h]\centering\small
\caption{Captured market share (\% of total demand) under the Huff competitive
model, vs.\ incumbent strength $C$ (Data source: Meituan).}
\label{tab:huff}
{\begin{tabular}{@{}c@{\ }c@{\ }c@{\ }c@{}}
\hline\up
$C$ & $p{=}5$ & $p{=}8$ & $p{=}12$ \\\hline\up
0.5 & 26.0\% & 31.9\% & 37.6\% \\
1.0 & 17.7\% & 24.2\% & 27.0\% \\
2.0 & 12.8\% & 14.6\% & 18.7\% \\
4.0 & 5.9\%  & 10.0\% & 12.4\% \\
8.0 & 3.4\%  & 6.4\%  & 7.8\% \down\\\hline
\end{tabular}}{}
\end{table}

\smallskip\noindent
The pattern tracks the theory: against a weak incumbent ($C{=}0.5$) additional
facilities steadily grow share ($26.0\%\to31.9\%\to37.6\%$), whereas against a
strong incumbent ($C{=}8$) the entrant is capped near $8\%$ regardless of $p$, and
the marginal gain of each added facility diminishes, confirming submodularity
(Theorem~\ref{thm:sub}) empirically. For market entry, this separates
contestable markets (low $C$), where investment in facilities pays off, from
entrenched ones (high $C$), where it does not.

\section{Discussion}\label{sec:discussion}

\subsection{Where does FBM-LNS fail?}\label{sec:disc-boundary}
Result~6 (\S\ref{res:step}) shows that FBM-LNS's advantage is specific to
\emph{smooth} (linear, exponential) decay, where the gradient/Weber structure of
\S\ref{sec:theory} is informative; under binary step decay the $(1{-}1/e)$
greedy---which depends only on the submodular marginal, not on gradients---matches
or exceeds it. We therefore recommend FBM-LNS for gradual coverage and the greedy
for binary coverage.

The boundary has a first-principles explanation consistent with the tractability
classification (Theorem~\ref{thm:tract}). The step function
$\mathbf{1}[r\le R_i]$ is piecewise constant, so $\nabla_{X_j}\phi_i$ is zero
almost everywhere, and the gradient-as-force identity (Theorem~\ref{thm:grad})
delivers no direction in which to move a facility; the relocate, by contrast,
depends only on the submodular marginal
$g(x)=\sum_i w_i[\phi_i(\|p_i-x\|)-c_i/w_i]^+$, which is well-defined for any
non-increasing $\phi_i$ including the step. The dichotomy is fundamental: the
\emph{continuous} signal (where to move a facility) requires smoothness, while
the \emph{combinatorial} signal (where to add one) does not. FBM-LNS exploits
both; at the step-decay boundary the gradient channel closes but the relocate
still provides a signal, so the method falls back to its greedy floor rather than
producing no output. A practical corollary is that for mixed objectives with a
smooth component (e.g., a ``soft'' step
$\phi_i(r)=\tfrac12[1+\cos(\pi r/R_i)]\mathbf{1}[r\le R_i]$) FBM-LNS remains
effective, because the smooth part restores the gradient signal.

\subsection{Runtime and scalability}\label{sec:runtime}
On the scalability benchmark (Figure~\ref{fig:scalability}), wall-clock times per
instance are under $1$--$3$\,s for the baselines and about $1$--$10$\,s for
FBM-LNS across $n=200$--$2000$ (with the lighter LNS budget used there); the
main-matrix runs use a heavier budget ($\approx 15$--$20$\,s per instance at
$n=300$--$1000$). The exact MIP adds $0.4$--$1.4$\,s for $n\le 500$. The additional runtime yields quality improvements that grow
with scale and heterogeneity (Tables~\ref{tab:main} and~\ref{tab:hetero}); where
runtime is critical, the LNS iteration budget can be reduced at a modest quality
cost. Figure~\ref{fig:scalability} traces quality and runtime up to $n=2000$,
$p=40$: FBM-LNS scales near-linearly in time and the quality advantage over the
$(1-1/e)$ greedy persists at the largest sizes.

The cost gap has a precise source. Each relocate iteration solves a
single-facility maximum-marginal subproblem on the residual (a weighted Weber
problem, $O(n\cdot|C|)$ for the candidate grid $C$) followed by a Lloyd polish
($O(np)$ per Lloyd step); over $L\approx 30$ iterations and $R\approx 6$ restarts
this dominates the baselines, which perform only a single Lloyd sweep. The
near-linear scaling in $n$ follows because both the Lloyd step and the marginal
solve are linear in $n$ and the number of effective iterations to convergence
does not grow with problem size---the partition structure stabilizes once each
facility commands a cell, regardless of how many points it contains. This is also
why reducing the LNS budget trades quality for time gracefully: each iteration is
an independent destroy-and-repair trial, so halving $L$ halves runtime at the cost
of fewer escapes from Lloyd fixed points.

\begin{figure}[htbp]
\FIGURE
{\includegraphics[width=0.92\linewidth]{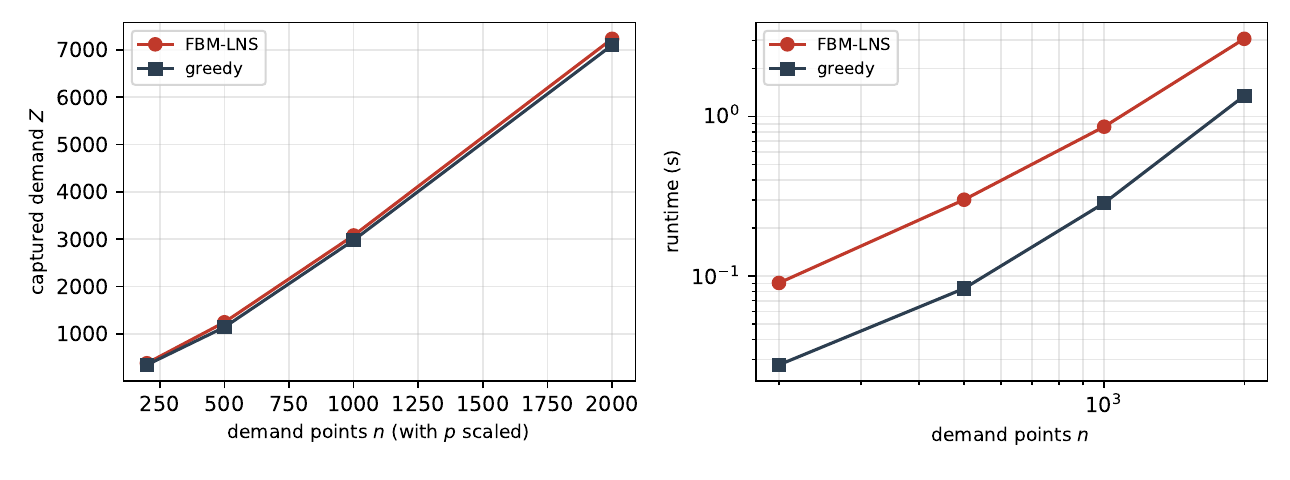}}
{Scalability on configuration TS1·V1 (conc./linear, $R_{cv}=0.5$). Left:
captured demand $Z$ (FBM-LNS above greedy throughout). Right: runtime (near-linear
in $n$ on the log axis); at $n=2000$, $p=40$ FBM-LNS finishes in about $10$\,s.\label{fig:scalability}}
{}
\end{figure}

\subsection{Can FBM-LNS be improved further?}\label{sec:disc-extend}
The relocate step selects which facility to move. The default rule
(least-contributing) is already near-optimal, which the ablation above shows is
not the source of the gain. We nonetheless tested whether a \emph{learned}
selector can improve on it: a linear policy $\mathrm{score}_j = w^\top
\mathrm{feat}_j$ over five per-facility features (contribution, cell size, mean
cell coverage, isolation, bias), trained by the cross-entropy method
(CEM; \citealp{DeBoerKroeseMannorRubinstein2005}) to
maximize captured demand on training instances, and evaluated on held-out test
instances (conc./linear, $n=300$, $p=5$, $K=10$). The learned policy attains
$Z=582.8$ versus $581.5$ for the heuristic (and $567.8$ for the greedy), a small
but statistically significant improvement ($+0.22\%$, $p=0.001$, $10/10$ wins).
The learned weights place large negative mass on the contribution feature, i.e.\
the optimizer rediscovers the relocate-the-low-contributor rule and adds a
marginal correction from the other features. This confirms that the
hand-designed selector already captures most of the selection signal, while a
learned policy extracts a small additional gain.

The small headroom is itself a consequence of the theory, not a failure of the
learning procedure. Result~4 showed that the destroy-and-repair \emph{structure}
accounts for the entire gain and that the selection rule is nearly immaterial: on
a Lloyd-stable configuration any facility removal exposes a comparable residual
marginal $g(x)$, so a linear policy over per-facility features operates on an
intrinsically low-information quantity and can recover little. (Indeed the
cross-entropy optimizer independently assigns large negative weight to the
contribution feature, rediscovering the least-contributor rule from data.) Closing
the remaining gap therefore likely requires policies that act on joint
destroy-and-repair moves or on the sequence of relocations, not on per-facility
scores.

\subsection{When should FBM-LNS be used?}
Table~\ref{tab:tradeoff} summarizes the trade-offs of FBM-LNS, so that
practitioners can judge its suitability for a given application.

\begin{table}[htbp]
\TABLE
{Strengths and costs of FBM-LNS.\label{tab:tradeoff}}
{\begin{tabular}{@{}p{0.46\linewidth}@{\quad}p{0.46\linewidth}@{}}
\hline\up
\textbf{Strengths} & \textbf{Costs and limitations} \\ \hline\up
Unified: one code handles $k$-means, $p$-median, gradual coverage, and shape
demand via a single switch.
& Slower than baselines ($\sim$10\,s vs.\ $1$--$3$\,s at $n{=}2000$) due to the LNS
relocate loop; reducible by lowering iteration budgets at a modest quality cost. \\[3pt]
Theory-grounded: every component maps to a proven property (Table~\ref{tab:design});
the greedy provides a $(1-1/e)$ lower-bound guarantee.
& No global-optimality guarantee for FBM-LNS itself; it is a heuristic whose
quality is validated empirically. \\[3pt]
Robust: significant superiority across 4 geographies, 3 decay families, 2
scales, and a heterogeneity sweep ($30/30$, $p{<}10^{-9}$).
& The advantage is specific to smooth (linear/exp) decay; for binary step
decay the greedy is competitive (\S\ref{sec:experiments}). \\[3pt]
Extensible: the framework accepts any non-increasing per-point function;
shape demand is handled by sampling.
& The learned selector (CEM) improves performance only marginally
($+0.22\%$); richer policies remain future work. \down\\ \hline
\end{tabular}}{}
\end{table}

In summary, the recommendation splits along the smooth/non-smooth line exposed in
\S\ref{sec:disc-boundary} and the structure/selection line exposed above. When
the decay is smooth and heterogeneous---the setting that defines V1--V2 and
motivates the per-point Weber weights of Theorem~\ref{thm:weber}---FBM-LNS is the
recommended method, because the gradient-as-force step is informative and
the relocate exploits the widening marginal gaps that heterogeneity creates
(Result~3). When a single implementation must also serve $k$-means, $p$-median,
or shape demand (V3--V5), the same code is competitive with bespoke SOTA
(Result~2). When the objective is binary (step decay), the gradient channel
closes and the greedy, which rides only the submodular marginal, is the
appropriate tool. When runtime dominates, halving the LNS budget recovers most of
the gain at half the cost (\S\ref{sec:runtime}), because each relocate iteration
is an independent trial whose marginal contribution diminishes. The unifying
principle is that FBM-LNS's two pillars map onto the two information
channels---continuous (gradient) and combinatorial (submodular marginal)---and
its advantage is largest exactly when both channels are open and most
differentially informative.

\section{Conclusion}\label{sec:conclusion}

We gave a structural theory for heterogeneous distance-decay facility location in
the continuous plane and a tractability classification that is the paper's
centerpiece: the discrete problem is always monotone submodular (so the
$(1-1/e)$ greedy guarantee holds regardless of decay shape or heterogeneity),
while the continuous cooperative objective is concave---and hence free of
local-optima traps---if and only if the decay is concave in the distance, a
dichotomy that identifies the clip $\max(0,\cdot)$ as a key mechanism that destroys
continuous concavity in common coverage specifications; the classification is tight:
concavity holds if and only if the decay is concave in distance. A tight LP relaxation makes the discretized problem
exactly solvable in seconds; our force-as-gradient / large-neighborhood-search
heuristic is within $0.5\%$ of the discrete optimum by fine-grid convergence and
outperforms the $(1-1/e)$ greedy, Cooper ALT, PSO, and weighted $k$-means across
distributions, decay families, scales, and heterogeneity levels; and on real
urban-delivery data, ignoring the density-derived heterogeneity costs measurable
demand and relocates facilities by up to $37\%$ of the map.

Three directions are left for
future work: a tighter continuous upper bound than the discretization bound of
Proposition~\ref{prop:cont_ub} (for example a conic relaxation of the
semi-infinite formulation exploiting the Weber substructure); a head-to-head
comparison against modern problem-specific solvers (variable neighborhood search
and mixed-integer second-order cone (MISOCP) methods for $p$-median, dedicated Huff-model solvers), which our
classical baselines do not include; and richer learned relocate policies, since
the cross-entropy selector of \S\ref{sec:disc-extend} already yields a small but
significant gain over the hand-designed rule, suggesting headroom for
contextual-bandit or attention-based selectors that act on joint destroy-and-repair
moves.

\section{Code and Data Disclosure}\label{sec:codedata}
The Python implementation of all algorithms studied in this paper---the proposed
FBM-LNS solver and every baseline (the $(1{-}1/e)$ greedy, multi-start Cooper
alternating location--allocation, particle swarm optimization, weighted
$k$-means, and the exact branch-and-bound and LP-bound procedures)---together
with the instance generators, the benchmark data sets, and a README giving full
reproduction instructions for every table and figure, are available at
\url{https://github.com/AgentLabCn/FBM-LNS}.
The synthetic instances are deterministic functions of a random seed and can be
regenerated exactly; the real-world delivery and retail data sets are the public
2026 Meituan Business Analytics Elite Challenge data and the retail panel
described in the real-world study. No data or code exemption is requested.

\appendix
\section{Proofs of Structural Results}\label{app:proofs}

This Electronic Companion gives the detailed, self-contained proofs of the
theorems, propositions, and lemmas stated in the main paper. We use the
following notation throughout: $d_i(\bm{X}) \coloneqq \min_{j \in J} \|p_i -
X_j\|$ is the nearest-facility distance to point $i$; $\Psi_i(\bm{X}) \coloneqq
w_i \max_{j \in J} \phi_i(\|p_i - X_j\|)$ is the captured demand at point $i$;
$K \coloneqq \operatorname{conv}\{p_i\}_{i \in I}$ is the demand convex hull;
and $\pi(i) \in \arg\min_{j \in J} \|p_i - X_j\|$ is $i$'s serving facility. We
assume each $\phi_i : \mathbb{R}_+ \to [0,1]$ is non-increasing with $\phi_i(0)
= 1$.

\begin{proof}{Proof of Theorem~\ref{thm:ndr} (Nearest-distance reduction).}
Fix $i \in I$ and write $d^\star \coloneqq \min_{j \in J} \|p_i - X_j\|$, with
the minimum attained at some $j^\star \in J$. For every $j \in J$ we have $\|p_i
- X_j\| \ge d^\star$, so by monotonicity of $\phi_i$, $\phi_i(\|p_i - X_j\|)
\le \phi_i(d^\star)$. Equality holds at $j = j^\star$, whence
\[
\max_{j \in J} \phi_i(\|p_i - X_j\|) = \phi_i(d^\star)
= \phi_i\!\Bigl(\min_{j \in J} \|p_i - X_j\|\Bigr),
\]
and the argmax is precisely $\arg\min_{j \in J} \|p_i - X_j\|$, i.e., $\pi(i) =
j^\star$. Multiplying by $w_i$ and summing over $i \in I$ yields $Z(\bm{X}) =
\sum_{i \in I} w_i\, \phi_i(d_i(\bm{X}))$. \Halmos
\end{proof}

\begin{proof}{Proof of Theorem~\ref{thm:hull} (Convex hull containment).}
$K$ is closed and convex. Suppose, for contradiction, that an optimal
$\bm{X}^\star$ has a facility $X_j^\star \notin K$, and let $X_j' \coloneqq
\mathrm{proj}_K(X_j^\star)$ be its metric projection onto $K$. By the projection
theorem for closed convex sets,
\[
\langle X_j^\star - X_j',\; p_i - X_j' \rangle \le 0 \qquad \text{for all } p_i
\in K.
\]
Decomposing $p_i - X_j^\star = (p_i - X_j') + (X_j' - X_j^\star)$ and expanding,
\begin{align*}
\|p_i - X_j^\star\|^2
&= \|p_i - X_j'\|^2 + \|X_j' - X_j^\star\|^2 \\
&\quad{} - 2\langle X_j^\star - X_j',\, p_i - X_j'\rangle \\
&\ge \|p_i - X_j'\|^2 + \|X_j' - X_j^\star\|^2 \\
&\ge \|p_i - X_j'\|^2,
\end{align*}
so $\|p_i - X_j^\star\| \ge \|p_i - X_j'\|$ for every $i$, and monotonicity of
$\phi_i$ gives $\phi_i(\|p_i - X_j^\star\|) \le \phi_i(\|p_i - X_j'\|)$.
Replacing $X_j^\star$ by $X_j'$ therefore does not decrease any $\Psi_i$ and
hence does not decrease $Z(\bm{X}^\star)$. Applying the same argument to every
facility outside $K$ yields an optimal solution in $K^p$. \Halmos
\end{proof}

\begin{proof}{Proof of Theorem~\ref{thm:sub} (Monotone submodularity) and the
curvature bound.}
Fix a finite candidate set $C$ and write $Z_D(S) \coloneqq \sum_{i \in I} w_i
\max_{c \in S} \phi_i(\|p_i - c\|)$ for $S \subseteq C$.

\emph{Step 1 (Monotonicity).} For $S \subseteq T \subseteq C$ and any $i$,
$\max_{c' \in S} \phi_i(\|p_i - c'\|) \le \max_{c' \in T} \phi_i(\|p_i -
c'\|)$; summing over $i$ gives $Z_D(S) \le Z_D(T)$.

\emph{Step 2 (Submodularity).} Let $S \subseteq T \subseteq C$ and $c \in C
\setminus T$. Define the pointwise marginal
\[
\Delta_i(S) \coloneqq w_i \Bigl[\phi_i(\|p_i - c\|) - \max_{c' \in S}
\phi_i(\|p_i - c'\|)\Bigr]^+.
\]
Because $\max_{c' \in S} \phi_i(\|p_i - c'\|) \le \max_{c' \in T} \phi_i(\|p_i -
c'\|)$, the bracketed term is non-increasing in the set, so $\Delta_i(S) \ge
\Delta_i(T)$. Summing over $i$,
\begin{align*}
Z_D(S \cup \{c\}) - Z_D(S)
&= \textstyle\sum_i \Delta_i(S)
\ge \textstyle\sum_i \Delta_i(T) \\
&= Z_D(T \cup \{c\}) - Z_D(T),
\end{align*}
which is the defining inequality of submodularity.

\emph{Step 3 (Curvature bound).} The total curvature of $Z_D$
\citep{ConfortiCornejols1984} is
\[
\kappa \;\coloneqq\; 1 - \min_{c \in C}\,
\frac{Z_D(C) - Z_D(C \setminus \{c\})}{Z_D(\{c\}) - Z_D(\emptyset)}
\;\in\; (0, 1],
\]
where the range uses monotonicity (Step~1) and submodularity (Step~2). The
classical result of \citet{ConfortiCornejols1984} then gives
$Z_D(S_{\mathrm{greedy}}) \ge \frac{1 - e^{-\kappa}}{\kappa}\, Z_D^\star$, and
$(1 - e^{-\kappa})/\kappa > 1 - 1/e$ whenever $\kappa < 1$. \Halmos
\end{proof}

\begin{proof}{Proof of Theorem~\ref{thm:weber} (Weber reduction) and
Corollaries~1--2.}
Fix a Voronoi cell $S \subseteq I$ and the linear decay $\phi_i(r) = [1 -
r/R_i]^+$. The within-cell single-facility objective is
\begin{align*}
\max_{X} \sum_{i \in S} w_i \Bigl(1 - \tfrac{\|p_i - X\|}{R_i}\Bigr)
&= \underbrace{\textstyle\sum_{i \in S} w_i}_{\text{constant}} \\
&\quad{} - \min_{X} \sum_{i \in S} \tfrac{w_i}{R_i}\,\|p_i - X\|.
\end{align*}
Setting $\alpha_i \coloneqq w_i / R_i$, the right-hand minimization is the
weighted Weber problem $\min_X \sum_{i \in S} \alpha_i \|p_i - X\|$, proving
Theorem~\ref{thm:weber}. \emph{Corollary~\ref{cor:hetero} (Heterogeneity):} the first-order optimality
condition is $\sum_{i \in S} \alpha_i (X - p_i) / \|X - p_i\| = \bm{0}$, so
smaller $R_i$ (larger $\alpha_i$) gives proportionally stronger pull.
\emph{Corollary~\ref{cor:uniq} (Uniqueness):} each $X \mapsto \|p_i - X\|$ is convex, so the
non-negative combination is convex; strict convexity of the Euclidean norm on
$\mathbb{R}^2$ except along rays through $p_i$ yields a unique minimizer unless
all $\{p_i\}_{i \in S}$ and $X$ are collinear. \Halmos
\end{proof}

\begin{proof}{Proof of Theorem~\ref{thm:grad} (Gradient-as-force).}
Write $Z(\bm{X}) = \sum_{i \in I} w_i \max_{j \in J} \psi_{ij}(\bm{X})$ with
$\psi_{ij}(\bm{X}) \coloneqq w_i \phi_i(\|p_i - X_j\|)$; each $\psi_{ij}$ is
$C^1$ in $X_j$ wherever $\phi_i$ is. Define the generic set
\[
\mathcal{G} \;\coloneqq\; \bigl\{\bm{X} : \pi(i) \text{ is unique for every } i
\bigr\};
\]
its complement is a finite union of Voronoi cell boundaries and hence has
Lebesgue measure zero. Fix $\bm{X} \in \mathcal{G}$. Because the active index
$\pi(i)$ is unique at $\bm{X}$, the pointwise maximum $\max_j \psi_{ij}$ is
differentiable there, and the chain rule for a maximum of $C^1$ functions with a
unique active index gives
\[
\nabla_{X_j} \max_{k \in J} \psi_{ik}(\bm{X}) = \begin{cases}
\nabla_{X_j} \psi_{i,\pi(i)}(\bm{X}), & j = \pi(i),\\
\bm{0}, & j \ne \pi(i).
\end{cases}
\]
Multiplying by $w_i$, summing over $i \in I$, and using $\nabla_{X_j}
\psi_{i,\pi(i)} = w_i \nabla_{X_j} \phi_i(\|p_i - X_j\|)$ yields the stated
identity in the main paper. \Halmos
\end{proof}

We now establish the two convex-analysis lemmas used in the tractability
classification, then prove the classification theorem.

\begin{lemma}[Diminishing increment]\label{lem:dimin}
If $g:\mathbb{R}_+\to\mathbb{R}$ is concave, then for $0\le s\le t$ and $a\ge
0$, $g(t+a)-g(t)\le g(s+a)-g(s)$.
\end{lemma}
\begin{proof}{Proof}
The secant slopes of a concave function are non-increasing: for $x_1\le x_2\le
x_3$, concavity gives
$\frac{g(x_2)-g(x_1)}{x_2-x_1}\ge\frac{g(x_3)-g(x_2)}{x_3-x_2}$.
Equivalently the one-sided derivative $g'_+$ is non-increasing, so
$g(x+a)-g(x)=\int_x^{x+a}g'_+(u)\,du$ is non-increasing in $x$; since $s\le t$,
the integral over $[t,t+a]$ is no larger than over $[s,s+a]$. Hence
$g(t+a)-g(t)\le g(s+a)-g(s)$. \Halmos
\end{proof}

\begin{lemma}[Concave-antitone composition]\label{lem:comp}
If $\phi:\mathbb{R}\to\mathbb{R}$ is concave and non-increasing and
$h:\mathbb{R}^d\to\mathbb{R}$ is convex, then $\phi\circ h$ is concave.
\end{lemma}
\begin{proof}{Proof}
For $\lambda\in[0,1]$, convexity of $h$ gives $h(\lambda x+(1-\lambda)y)\le
\lambda h(x)+(1-\lambda)h(y)$. Since $\phi$ is non-increasing, applying $\phi$
preserves the inequality:
\[
\phi\bigl(h(\lambda x+(1-\lambda)y)\bigr)\;\ge\;\phi\bigl(\lambda
h(x)+(1-\lambda)h(y)\bigr).
\]
Concavity of $\phi$ then gives $\phi(\lambda h(x)+(1-\lambda)h(y))\ge
\lambda\phi(h(x))+(1-\lambda)\phi(h(y))$. Chaining the two inequalities yields
$\phi(h(\lambda x+(1-\lambda)y))\ge\lambda\phi(h(x))+(1-\lambda)\phi(h(y))$,
the defining inequality of concavity for $\phi\circ h$. \Halmos
\end{proof}

\begin{proof}{Proof of Theorem~\ref{thm:tract} (Tractability classification).}
\emph{(i) Discrete submodularity.} The non-cooperative case is Theorem~\ref{thm:sub}. For
the cooperative case, write $Z_D^g(S)=\sum_i w_i g(\sigma_i(S))$ with
$\sigma_i(S)=\sum_{c\in S}\phi_i(\|p_i-c\|)$; $\sigma_i$ is modular and, since
each $\phi_i(\|p_i-c\|)\ge 0$, non-decreasing in $S$. For $S\subseteq T$ and
$c\notin T$, set $a_i\coloneqq\phi_i(\|p_i-c\|)\ge 0$; the per-point marginal is
$\Delta_i^g(S)=w_i[g(\sigma_i(S)+a_i)-g(\sigma_i(S))]$. Since
$\sigma_i(S)\le\sigma_i(T)$, Lemma~1 gives
$\Delta_i^g(S)\ge\Delta_i^g(T)$; summing yields submodularity. Monotonicity
follows from $g$ non-decreasing and $a_i\ge 0$.

\emph{(ii,$\Leftarrow$) Cooperative concavity.} Assume each $\phi_i$ is concave
on $\mathbb{R}_+$ and non-increasing, $g$ is concave non-decreasing, and $w_i\ge
0$. The map $X_j\mapsto\|p_i-X_j\|$ is convex; by Lemma~2,
$X_j\mapsto\phi_i(\|p_i-X_j\|)$ is concave. As a function of $\bm{X}$, each term
$\phi_i(\|p_i-X_j\|)$ is concave in $X_j$ and constant in $X_k$ for $k\ne j$,
hence concave in $\bm{X}$; the sum $U_i=\sum_j\phi_i(\|p_i-X_j\|)$ is concave,
and the composition $g\circ U_i$ is concave (concave non-decreasing composed with
concave). The non-negative weighted sum $Z^g=\sum_i w_i\,g(U_i)$ is concave.
Under concavity every local maximum is global, and projected gradient ascent
converges.

\emph{(ii,$\Rightarrow$) Necessity (contrapositive).} If some $\phi_{i_0}$ is
not concave on $[0,M]$, there exist $r_1<r_2<r_3$ in $[0,M]$ violating the
concavity inequality. Set $p=1$, $n=1$, $g=\mathrm{id}$, $w_{i_0}=1$; then
$Z(X)=\phi_{i_0}(\|p_{i_0}-X\|)$. Along the ray $X=p_{i_0}+t\,u$
($\|u\|=1$), $Z=\phi_{i_0}(t)$, which is not concave. Hence $Z^g$ is not
concave.

\emph{(iii) Non-cooperative non-concavity for $p\ge 2$.} Take one demand point
at the origin with $w=1$, $p=2$ facilities, and any non-constant non-increasing
$\phi$ with $\phi(0)=1$; choose $a>0$ with $\phi(a/2)<1$. Consider
$\bm{X}=((a,0),(0,0))$: facility~2 sits at the point, so
$Z(\bm{X})=\max(\phi(a),\phi(0))=1$. Likewise $Z(\bm{Y})=1$ for
$\bm{Y}=((0,0),(a,0))$. Their midpoint
$\bm{M}=\tfrac12(\bm{X}+\bm{Y})=((a/2,0),(a/2,0))$ places both facilities at
distance $a/2$, so $Z(\bm{M})=\phi(a/2)<1=(Z(\bm{X})+Z(\bm{Y}))/2$, violating
concavity at the midpoint. \Halmos
\end{proof}

\begin{proof}{Proof of Proposition~\ref{prop:lloyd} (Lloyd fixed point).}
At a fixed point $\bm{X}^\star$ of the alternating map, the location step makes
no update, so each $X_j^\star$ satisfies
\[
X_j^\star \in \arg\max_{X \in \mathbb{R}^2}\,
\sum_{i:\,\pi(i)=j} w_i\, \phi_i(\|p_i - X\|),
\]
equivalently $0 \in \partial_{C, X_j} Z$ (Clarke subdifferential with respect to
the single-facility coordinate). The allocation step makes no reassignment, so
$\pi(i) = \arg\min_{j} \|p_i - X_j^\star\|$ for every $i$, confirming Voronoi
stability. Together, the two conditions state that no single facility can be
moved to strictly increase $Z$ while the others and the assignment are held
fixed; this is coordinate-wise local optimality. \Halmos
\end{proof}

\begin{proof}{Proof of Proposition~\ref{prop:reloc} (Monotone relocate).}
Let $\{\bm{X}^{(t)}\}_{t \ge 0}$ be the sequence generated by the relocate step
under strict-improvement acceptance, $Z(\bm{X}^{(t+1)}) > Z(\bm{X}^{(t)})$. Then
$\{Z(\bm{X}^{(t)})\}$ is strictly increasing. Once the relocate draws its
candidate from a finite set $C$, the configuration space $C^p$ is finite, so a
strictly increasing sequence on the finite set $\{Z(\bm{X}) : \bm{X} \in C^p\}$
terminates in finitely many steps. At the terminal configuration $\hat{\bm{X}}$:
(i) the Lloyd location step has converged, so $\hat{\bm{X}}$ is Lloyd-stable by
Proposition~\ref{prop:lloyd}; and (ii) no single-facility re-placement from $C$ strictly
increases $Z$, by the strict-improvement acceptance rule. \Halmos
\end{proof}

\begin{proof}{Proof of Proposition~\ref{prop:cont_ub} (Continuous upper bound via grid density).}
Let $C$ be a grid of spacing $h$ covering $K=\operatorname{conv}\{p_i\}$, so
for every $x\in K$ there is a grid node $\tilde x\in C$ with $\|x-\tilde x\|\le
h/\sqrt{2}$. Fix any continuous configuration $\bm{X}=(X_1,\dots,X_p)$ and let
$\tilde X_j\in C$ be a nearest grid node to $X_j$. For each demand point $i$,
write $d_{ij}=\|p_i-X_j\|$ and $\tilde d_{ij}=\|p_i-\tilde X_j\|$. The bound
$|\max_j a_j-\max_j b_j|\le\max_j|a_j-b_j|$, the $\mathrm{Lip}(\phi_i)$-Lipschitz
continuity of $\phi_i$ (for linear and exponential decay
$\mathrm{Lip}(\phi_i)=1/R_i$), and $\max_j\|X_j-\tilde X_j\|\le h/\sqrt{2}$ give
\[
\bigl|\max_j\phi_i(d_{ij})-\max_j\phi_i(\tilde d_{ij})\bigr|
\;\le\; \mathrm{Lip}(\phi_i)\,h/\sqrt{2}.
\]
Multiplying by $w_i$ and summing over $i$ gives
$Z(\bm{X})\le Z(\tilde{\bm{X}})+(h/\sqrt{2})L_\phi$ with
$L_\phi=\sum_i w_i\,\mathrm{Lip}(\phi_i)$. Since $\tilde{\bm{X}}\in C^p$,
$Z(\tilde{\bm{X}})\le Z_{\mathrm{grid\text{-}MIP}}(C)$, so
$Z(\bm{X})\le Z_{\mathrm{grid\text{-}MIP}}(C)+(h/\sqrt{2})L_\phi$ for every
$\bm{X}$; taking the supremum over $\bm{X}$ yields
$Z^*\le Z_{\mathrm{grid\text{-}MIP}}(C)+(h/\sqrt{2})L_\phi$. Finally
$Z_{\mathrm{grid\text{-}MIP}}(C)\le\mathrm{LP}(C)$ because the LP relaxes the
integer program. \Halmos
\end{proof}

\begin{proof}{Proof of Proposition~\ref{prop:coop_ub} (Grid-free continuous upper bound for concave decays).}
For each $i$ and $\bm{X}$,
$\max_j\phi_i(\|p_i-X_j\|)\le\sum_j\phi_i(\|p_i-X_j\|)$ since $\phi_i\ge0$.
Weighting by $w_i$ and summing,
$Z(\bm{X})\le\sum_i w_i\sum_j\phi_i(\|p_i-X_j\|)
=\sum_j\bigl(\sum_i w_i\phi_i(\|p_i-X_j\|)\bigr)$.
The bracketed term depends only on $X_j$, so the maximum over $\bm{X}\in K^p$
separates and equals
$p\cdot\max_{x\in K}\sum_i w_i\phi_i(\|p_i-x\|)$. For a decay that is concave
and non-increasing in the distance (unclipped linear $\phi_i(r)=1-r/R_i$ and
quadratic $\phi_i(r)=1-(r/R_i)^2$ are both concave in $r$), Lemma~\ref{lem:comp} gives that
$\phi_i(\|p_i-x\|)$ is concave in $x$; hence
$f(x)=\sum_i w_i\phi_i(\|p_i-x\|)$ is concave. A concave function attains its
maximum over the compact convex hull $K$ at an extreme point, i.e., a demand
point $p_k$. Hence the bound is $p\cdot\max_k\sum_i
w_i\phi_i(\|p_i-p_k\|)$, computable in $O(n^2)$. \Halmos
\end{proof}
\section{Abbreviations}\label{app:abbrev}
Table~\ref{tab:abbrev} lists every abbreviation used in the paper, its expansion,
and the section or table where it first appears. Each abbreviation is also
expanded at its first occurrence in the running text.

\begin{table}[h]
\centering
\TABLE
{Abbreviations used in the paper.\label{tab:abbrev}}
{\begin{tabular}{@{}l@{\quad}p{0.55\linewidth}@{\quad}l@{}}
\hline\up
Abbrev. & Expansion & First use \\ \hline\up
ALT     & alternating location--allocation                                       & \S\ref{sec:intro} \\
CEM     & cross-entropy method                                                   & \S\ref{sec:disc-extend} \\
CI      & confidence interval                                                    & Table~\ref{tab:stats} \\
FBM     & force-based metaheuristic                                              & \S\ref{sec:algo} \\
FBM-LNS & force-based metaheuristic with large-neighborhood search               & \S\ref{sec:contrib} \\
FDS     & finite dominating set                                                  & \S\ref{sec:difficulty} \\
LP      & linear programming                                                     & \S\ref{sec:contrib} \\
LNS     & large-neighborhood search                                              & \S\ref{sec:algo} \\
MCLP    & maximum covering location problem                                      & \S\ref{sec:examples} \\
MGCLP   & multiple gradual cover location problem                                & \S\ref{sec:related} \\
MIP     & mixed-integer program                                                  & \S\ref{sec:contrib} \\
PSO     & particle swarm optimization                                            & \S\ref{sec:specialize} \\
$R_{cv}$ & coefficient of variation of the decay scale $\{R_i\}$                 & \S\ref{sec:setup} \\
RMSE    & root-mean-square error                                                 & \S\ref{sec:case} \\
SOCP    & second-order cone program                                              & \S\ref{sec:tractability} \down\\ \hline
\end{tabular}}{}
\end{table}

\ACKNOWLEDGMENT{Omitted for double-anonymous review.}

\bibliographystyle{informs2014}
\bibliography{refsOR}

\begin{thebibliography}{26}
\providecommand{\natexlab}[1]{#1}
\providecommand{\url}[1]{\texttt{#1}}
\providecommand{\urlprefix}{URL }

\bibitem[{{\'Alvarez-Miranda} \protect\BIBand{} Sinnl(2019)}]{AlvarezMirandaSinnl2019}
{\'Alvarez-Miranda} E, Sinnl M (2019) An exact solution framework for the multiple gradual cover location problem. \emph{Computers \& Operations Research} 108:82--96, \urlprefix\url{http://dx.doi.org/10.1016/j.cor.2019.04.003}.

\bibitem[{Arthur \protect\BIBand{} Vassilvitskii(2007)}]{ArthurVassilvitskii2007}
Arthur D, Vassilvitskii S (2007) {$k$}-means${}^{++}$: The advantages of careful seeding. \emph{Proceedings of the 18th Annual ACM-SIAM Symposium on Discrete Algorithms (SODA)}, 1027--1035, \urlprefix\url{http://dx.doi.org/10.1145/1283383.1283494}.

\bibitem[{Bansal \protect\BIBand{} Kianfar(2017)}]{BansalKianfar2017}
Bansal M, Kianfar K (2017) Planar maximum coverage location problem with partial coverage and rectangular demand and service zones. \emph{INFORMS Journal on Computing} 29(1):152--169, \urlprefix\url{http://dx.doi.org/10.1287/ijoc.2016.0722}.

\bibitem[{Bansal \protect\BIBand{} Shojaee(2020)}]{BansalShojaee2020}
Bansal M, Shojaee P (2020) Planar maximum coverage location problem with partial coverage, continuous spatial demand, and adjustable quality of service, optimization-Online technical report 8178.

\bibitem[{Berman et~al.(2009)Berman, Kalcsics, Krass, \protect\BIBand{} Nickel}]{BermanKalcsicsKrassNickel2009}
Berman O, Kalcsics J, Krass D, Nickel S (2009) The ordered gradual covering location problem on a network. \emph{Discrete Applied Mathematics} 157(18):3689--3707, \urlprefix\url{http://dx.doi.org/10.1016/j.dam.2009.08.003}.

\bibitem[{Berman et~al.(2003)Berman, Krass, \protect\BIBand{} Drezner}]{BermanKrassDrezner2003}
Berman O, Krass D, Drezner Z (2003) The gradual covering decay location problem on a network. \emph{European Journal of Operational Research} 151(3):474--480, \urlprefix\url{http://dx.doi.org/10.1016/S0377-2217(02)00604-5}.

\bibitem[{Church \protect\BIBand{} Murray(2018)}]{ChurchMurray2018}
Church RL, Murray AT (2018) \emph{Location Covering Models: History, Applications and Advancements} (Springer), \urlprefix\url{http://dx.doi.org/10.1007/978-3-319-99846-6}.

\bibitem[{Conforti \protect\BIBand{} Cornu\'ejols(1984)}]{ConfortiCornejols1984}
Conforti M, Cornu\'ejols G (1984) Submodular set functions, matroids and the greedy algorithm: Tight worst-case bounds and some generalizations of the {Rado--Edmonds} theorem. \emph{Discrete Applied Mathematics} 7(3):251--274, \urlprefix\url{http://dx.doi.org/10.1016/0166-218X(84)90003-9}.

\bibitem[{Cooper(1964)}]{Cooper1964}
Cooper L (1964) Heuristic methods for location-allocation problems. \emph{SIAM Review} 6(1):37--53, \urlprefix\url{http://dx.doi.org/10.1137/1006005}.

\bibitem[{De~Boer et~al.(2005)De~Boer, Kroese, Mannor, \protect\BIBand{} Rubinstein}]{DeBoerKroeseMannorRubinstein2005}
De~Boer PT, Kroese DP, Mannor S, Rubinstein RY (2005) A tutorial on the cross-entropy method. \emph{Annals of Operations Research} 134(1):19--67, \urlprefix\url{http://dx.doi.org/10.1007/s10479-005-5724-z}.

\bibitem[{Drezner(1994)}]{Drezner1994}
Drezner T (1994) Optimal continuous location of a retail facility, facility attractiveness, and market share: An interactive model. \emph{Journal of Retailing} 70(1):49--64, \urlprefix\url{http://dx.doi.org/10.1016/0022-4359(94)90028-0}.

\bibitem[{Drezner et~al.(2004)Drezner, Wesolowsky, \protect\BIBand{} Drezner}]{DreznerWesolowskyDrezner2004}
Drezner Z, Wesolowsky GO, Drezner T (2004) The gradual covering problem. \emph{Naval Research Logistics} 51(6):841--855, \urlprefix\url{http://dx.doi.org/10.1002/nav.20021}.

\bibitem[{Feige(1998)}]{Feige1998}
Feige U (1998) A threshold of $\ln n$ for approximating set cover. \emph{Journal of the ACM} 45(4):634--652, \urlprefix\url{http://dx.doi.org/10.1145/285055.285059}.

\bibitem[{Fr\"anti \protect\BIBand{} Sieranoja(2018)}]{FrantiSieranoja2018}
Fr\"anti P, Sieranoja S (2018) K-means properties on six clustering benchmark datasets. \emph{Applied Intelligence} 48(12):4743--4759, \urlprefix\url{http://dx.doi.org/10.1007/s10489-018-1238-7}.

\bibitem[{Hansen \protect\BIBand{} Mladenovi{\'c}(1997)}]{HansenMladenovic1997}
Hansen P, Mladenovi{\'c} N (1997) Variable neighborhood search for the {$p$}-median. \emph{Location Science} 5(4):207--226, \urlprefix\url{http://dx.doi.org/10.1016/S0966-8349(98)00030-8}.

\bibitem[{Huangfu \protect\BIBand{} Hall(2018)}]{HuangfuHall2018}
Huangfu Q, Hall JAJ (2018) Parallelizing the dual revised simplex method. \emph{Mathematical Programming Computation} 10(1):119--142, \urlprefix\url{http://dx.doi.org/10.1007/s12532-017-0130-5}.

\bibitem[{Huff(1964)}]{Huff1964}
Huff DL (1964) Defining and estimating a trading area. \emph{Journal of Marketing} 28(3):34--38, \urlprefix\url{http://dx.doi.org/10.1177/002224296402800307}.

\bibitem[{Karasakal \protect\BIBand{} Karasakal(2004)}]{KarasakalKarasakal2004}
Karasakal O, Karasakal EK (2004) A maximal covering location model in the presence of partial coverage. \emph{Computers \& Operations Research} 31(9):1515--1526, \urlprefix\url{http://dx.doi.org/10.1016/S0305-0548(03)00105-9}.

\bibitem[{Kennedy \protect\BIBand{} Eberhart(1995)}]{KennedyEberhart1995}
Kennedy J, Eberhart R (1995) Particle swarm optimization. \emph{Proceedings of the IEEE International Conference on Neural Networks (ICNN'95)}, volume~4, 1942--1948, \urlprefix\url{http://dx.doi.org/10.1109/ICNN.1995.488968}.

\bibitem[{Lloyd(1982)}]{Lloyd1982}
Lloyd SP (1982) Least squares quantization in {PCM}. \emph{IEEE Transactions on Information Theory} 28(2):129--137, \urlprefix\url{http://dx.doi.org/10.1109/TIT.1982.1056489}.

\bibitem[{Lu \protect\BIBand{} Zhou(2016)}]{LuZhou2016}
Lu Y, Zhou HH (2016) Statistical and computational guarantees of {Lloyd's} algorithm and its variants, arXiv:1612.02099.

\bibitem[{Nemhauser et~al.(1978)Nemhauser, Wolsey, \protect\BIBand{} Fisher}]{NemhauserWolseyFisher1978}
Nemhauser GL, Wolsey LA, Fisher ML (1978) An analysis of approximations for maximizing submodular set functions---{I}. \emph{Mathematical Programming} 14(1):265--294, \urlprefix\url{http://dx.doi.org/10.1007/BF01588971}.

\bibitem[{Ostrovsky et~al.(2013)Ostrovsky, Rabani, Schulman, \protect\BIBand{} Swamy}]{OstrovskyRabaniSchulmanSwart2012}
Ostrovsky R, Rabani Y, Schulman LJ, Swamy C (2013) The effectiveness of {Lloyd}-type methods for the {$k$}-means problem. \emph{Journal of the ACM} 59(6):Article 28, \urlprefix\url{http://dx.doi.org/10.1145/2395116.2395117}.

\bibitem[{Reinelt(1991)}]{Reinelt1991}
Reinelt G (1991) {TSPLIB}---a traveling salesman problem library. \emph{ORSA Journal on Computing} 3(4):376--384, \urlprefix\url{http://dx.doi.org/10.1287/ijoc.3.4.376}.

\bibitem[{Weiszfeld(1937)}]{Weiszfeld1937}
Weiszfeld E (1937) Sur le point pour lequel la somme des distances de $n$ points donn\'es est minimum. \emph{Tohoku Mathematical Journal} 43:355--386.

\bibitem[{Yang et~al.(2019)Yang, Fang, Xu, Yin, Li, \protect\BIBand{} Lu}]{YangFangXuYinLiLu2019}
Yang X, Fang Z, Xu Y, Yin L, Li J, Lu S (2019) Spatial heterogeneity in spatial interaction of human movements---insights from large-scale mobile positioning data. \emph{Journal of Transport Geography} 78:29--40, \urlprefix\url{http://dx.doi.org/10.1016/j.jtrangeo.2019.05.010}.

\end{thebibliography}

\end{document}